\documentclass[showpacs,twocolumn,aps,pra,longbibliography,superscriptaddress,notitlepage]{revtex4-2}
\usepackage{qcircuit}
\usepackage[dvips]{graphicx}
\usepackage{amsmath,amssymb,amsthm,mathrsfs,amsfonts,dsfont}
\usepackage{subfigure, epsfig}
\usepackage{braket}
\usepackage{bm}
\usepackage{enumerate}
\usepackage{algorithm}
\usepackage{algpseudocode}
\usepackage{diagbox}
\usepackage{extarrows}
\usepackage{physics}
\usepackage{color}
\usepackage{multirow}
\usepackage[marginal]{footmisc}
\usepackage{comment}
\usepackage{tikz}
\usepackage{gensymb}
\usepackage[]{qcircuit}
\usetikzlibrary{arrows}
\usetikzlibrary{shapes,fadings,snakes}
\usetikzlibrary{decorations.pathmorphing,patterns}
\usetikzlibrary{calc}
\usetikzlibrary{positioning}
\usepackage{hyperref}

\graphicspath{{./}}

\hypersetup{colorlinks=true, linkcolor=blue, citecolor=blue, urlcolor=blue}

\usepackage{comment}
\usepackage{dcolumn}
\usepackage{bm}
\usepackage{tikz}
\usepackage{ulem}
\usepackage{graphicx}
\usepackage{placeins}

\usepackage{varioref}

\labelformat{section}{#1}

\labelformat{figure}{#1}

\labelformat{proposition}{#1}

\labelformat{lemma}{#1}

\labelformat{theorem}{#1}

\labelformat{observation}{#1}

\labelformat{definition}{#1}

\labelformat{corollary}{#1}

\labelformat{problem}{#1}

\newtheorem{theorem}{Theorem}

\graphicspath{{Images/}}
\usepackage{ulem}

\begin{document}

\title{Memory dimension detection in open quantum dynamics via pseudo-control}

\let\oldaddcontentsline\addcontentsline
\renewcommand{\addcontentsline}[3]{}

\begin{abstract}
Characterising open quantum dynamics requires identifying the environmental degrees of freedom that shape the system's evolution. The memory dimension, defined as the size of the smallest effective memory needed to reproduce this influence, quantifies the dynamically relevant environmental resources and provides a common basis for comparing different many-body environments. Here we introduce a pseudo-control protocol that encodes information about the memory dimension in an experimentally accessible interference matrix, whose rank certifies a lower bound on this dimension without direct access to the environment or full process-tensor tomography. Applied to quantum many-body environments, the protocol allows flexible adjustment of the interaction and interval times, revealing how these parameters, environmental dynamics, and system-environment interactions shape detectable memory.
Our approach provides a new operational technique for characterising open quantum dynamics and probing the structure of system-environment interactions through their memory signatures.
\end{abstract}

\author{Chang Liu}
\email{harperlcq@gmail.com}
\affiliation{Yau Mathematical Sciences Center, Tsinghua University, Beijing 100084, China}

\author{Alexander Yosifov}
\affiliation{Clarendon Laboratory, University of Oxford, Parks Road, Oxford, OX1 3PU, United Kingdom}

\author{Ximing Wang}
\affiliation{Nanyang Quantum Hub, School of Physical and Mathematical Sciences, Nanyang Technological University, Singapore 639798, Singapore}

\author{Zhenhuan Liu}
\thanks{\href{mailto:qubithuan@gmail.com}{qubithuan@gmail.com}}
\affiliation{Quantum Research Center, Technology Innovation Institute (TII), Abu Dhabi, United Arab Emirates}

\author{Jinzhao Sun}
\email{jinzhao.sun.phys@gmail.com}
\affiliation{School of Physical and Chemical Sciences, Queen Mary University of London, London, E1 4NS, United Kingdom}

\maketitle

The spatio-temporal structure of the environment plays a central role in open quantum dynamics~\cite{Breuer2016NonMarkovian,deVega2017Dynamics}.
Environmental memory offers a signature of this structure: information leaked from past interactions can influence the system's subsequent evolution~\cite{Pollock2018OperationalMarkov,Pollock2018ProcessTensor}.
The resources required to retain this information are captured by the effective memory dimension---the minimum Hilbert-space dimension of a memory needed to model the process~\cite{Bisio2012MemoryCost,Budroni2019MemoryCost}.
This notion connects the physical structure of the environment to the resources required for simulation and control.
In many-body systems, memory reflects information spreading under internal dynamics and helps determine the resources needed to describe local relaxation and simulate quantum many-body dynamics~\cite{Lerose2021InfluenceMatrix,Lerose2023SpaceTimeDuality}.
Related memory costs arise in simulations of general open quantum dynamics~\cite{Luchnikov2019SimulationComplexity,Strathearn2018,Cygorek2022EnvironmentCompression} and in models for characterising and controlling temporally correlated noise~\cite{White2020,Fux2021OptimalControl,Butler2024Optimizing,White2025UnifiedCharacterization}.
Comparing these resource requirements across quantum and classical models also reveals advantages of quantum memory: certain stochastic processes can be simulated using a smaller memory dimension~\cite{Ghafari2019Dimensional,Elliott2020ExtremeDimensionality} or with higher accuracy~\cite{yang2023provably} in quantum models than in classical ones.
At the operational level, environmental memory also affects the performance of quantum communication~\cite{BowenMancini2004,KretschmannWerner2005,Mele2024OpticalFibersMemory,Cimini2020CorrelatedNoiseCapacity,Pirandola2021EnvironmentAssisted} and the requirements for fault-tolerant quantum computation~\cite{Bombin2016}.

In real-world studies, however, inferring the memory dimension is a complicated task.
Environmental memory resides in degrees of freedom that are usually inaccessible, so its dimension must be inferred indirectly from the system's interactions with the environment.
Existing methods infer properties of this memory, including its presence and the magnitude of associated non-Markovianity \cite{Wolf2008Assessing,Breuer2009Measure,Rivas2010Entanglement,Pollock2018OperationalMarkov,Goswami2021ExperimentalProcess,das2026hysteretic}, whether observed temporal statistics are classical \cite{Milz2020ClassicalMemory}, and whether the memory itself must be quantum \cite{Giarmatzi2021Witnessing,Backer2024LocalDisclosure,Backer2025EntropicWitness,yosifov2025emergence}.
Quantum combs and process tensors describe the system's response to arbitrary sequences of interventions \cite{Chiribella2008QuantumCircuit,Chiribella2009QuantumNetworks,Pollock2018ProcessTensor}; process-tensor tomography reconstructs this complete multi-step input-output map \cite{White2022ProcessTomography,Giarmatzi2025MultitimeTomography}, while structured methods learn compact representations of it \cite{White2025UnifiedCharacterization}.
Together, these approaches motivate a more focused question about memory dimension: can measurements of accessible outputs certify a lower bound on the memory dimension carried between interactions without reconstructing the complete multi-step map?


\begin{figure}[bp]
    \centering
    \includegraphics[width=0.89\linewidth]{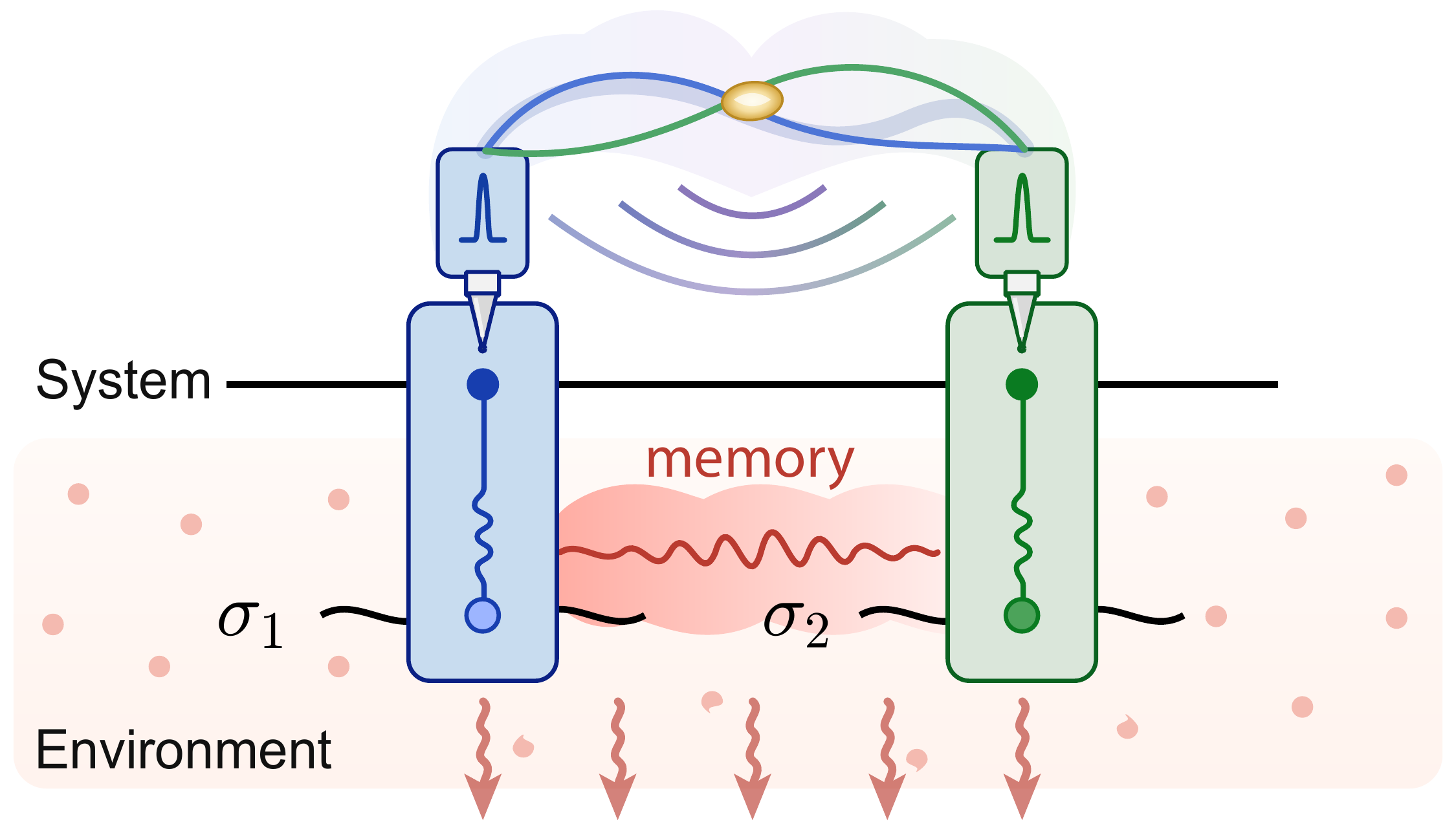}
    \caption{Schematic of the two-interaction process. The red region represents memory retained from the first interaction that can influence the second.}
    \label{fig:memory_setting}
\end{figure}


In this work, we introduce a pseudo-control protocol in which a system probes the same environment in two interactions separated by a waiting interval, Fig.~\ref{fig:memory_setting}.
During the interval, the system is decoupled while the environment evolves and may retain information from the first interaction that influences the second.
By coherently comparing the two probing histories, the protocol encodes this temporal dependence in an interference matrix $M$, accessible through measurements on the control and retained outputs without direct access to the environment.
We prove that if $\operatorname{rank}(M)=r$, every compatible realisation of the two-interaction process requires a memory of dimension at least $\lceil\sqrt{r}\rceil$.
The interaction durations and the waiting interval can be adjusted to examine how the detectable memory dimension changes as information is exchanged with and evolves within the environment.
A reference system can further extend the range of memory dimensions detectable by the protocol.

We use this flexibility to study dissipative and nondissipative many-body environments with different internal dynamics.
{We find that the interaction and interval times affect how detectable memory builds up and changes as the environment evolves.} 
Varying the coupling geometry and coupling strength produces distinct rank signatures: at strong coupling, the rank decreases for a boundary geometry but not for a ladder geometry, even with the same internal environmental dynamics.
These signatures show that certified bounds on memory dimension can reveal information about the many-body environment.

\textit{Protocol and memory dimension bound.---}
We start with illustrating our protocol in Fig.~\ref{fig:pseudocontrol} using a simplified memory model for easier comprehension.
The system $S$ is idle between the interactions $\mathcal C_1$ and $\mathcal C_2$.
After $\mathcal C_1$, an environmental subsystem is passed to $\mathcal C_2$, while other outputs are discarded; $\mathcal C_2$ can also receive a fresh environmental input in state $\sigma_2$.
The environment $E$ can thereby exchange information with additional degrees of freedom, although its detailed evolution is left unspecified in this model.
The memory dimension $d_{\mathrm m}$ is the smallest Hilbert-space dimension among all compatible realisations, shown by the red wire in a minimal realisation.
Degrees of freedom used and discarded within either channel do not count towards $d_{\mathrm m}$.

We use pseudo-control to turn the correlation between the two interaction times into a dependence between the rows and columns of an interference matrix at the output~\cite{nakayama2013universal}.
The circuit coherently routes $S$ through one interaction or the other, following earlier interferometric and coherent-control schemes \cite{Oi2003Interference,Chiribella2019Trajectories,Abbott2020CoherentControl,Kristjansson2021Latent}.
The control qubit $c$ is prepared in $\ket{+}$, while $S$ and an isolated reference $R$ are prepared in a pure state $\ket{\Psi}_{RS}$, with $\rho_S=\operatorname{Tr}_R[\ket{\Psi}\bra{\Psi}_{RS}]$.
Controlled \texttt{SWAP}s send $S$ through $\mathcal C_1$ on the $\ket{0}_c$ branch and through $\mathcal C_2$ on the $\ket{1}_c$ branch; the bypassed interaction receives the fixed input $\rho_2$ or $\rho_1$, respectively.

After all other outputs are traced out, the state retained on $cRS$ is
\begin{equation}
\omega_{cRS}
=
\sum_{i,j=0}^{1}
\ket{i}\!\bra{j}_c\otimes\omega_{ij}^{RS}.
\label{eq:control_block_decomposition}
\end{equation}
The off-diagonal control block $M_{RS}:=\omega_{01}^{RS}$ contains the interference between the two routes; its rank gives the memory-dimension bound stated below.

\begin{figure}[htbp]
    \centering
    \includegraphics[width=0.92\linewidth]{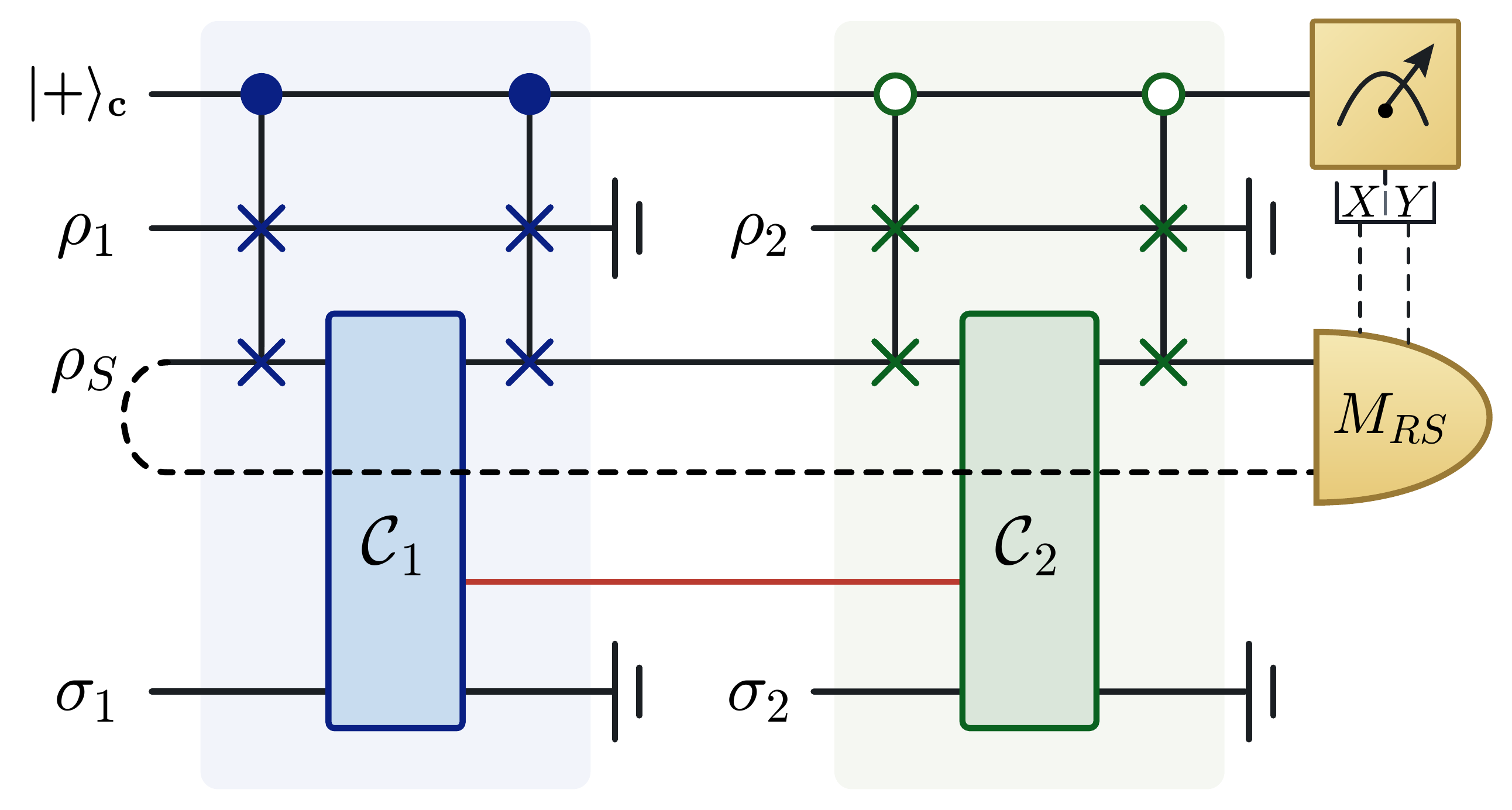}
    \caption{Pseudo-control circuit for detecting memory dimension. From top to bottom, the wires show the control $c$, auxiliary system inputs $\rho_1,\rho_2$, the input system $S$ with its isolated reference $R$ (dashed), and local environmental inputs $\sigma_1,\sigma_2$. The blue and green boxes denote interaction channels $\mathcal C_1$ and $\mathcal C_2$, respectively, while the red wire carries memory with $\dim\mathcal H_m=d_{\mathrm m}$. Filled and open circles indicate \texttt{SWAP}s controlled on $\ket{1}_c$ and $\ket{0}_c$, respectively; terminal bars mark discarded outputs. The symbols at right indicate $X/Y$ measurements on $c$ and measurement of $RS$ to reconstruct $M_{RS}$.
}
    \label{fig:pseudocontrol}
\end{figure}

\begin{theorem}[Memory-dimension lower bound]
\label{thm:memory_bound}
Let $\ket{\Psi}_{RS}$ be a pure reference-system input, and define $M_{RS}:=\omega_{01}^{RS}$ as the off-diagonal control block in Eq.~\eqref{eq:control_block_decomposition}.
If the two-interaction process has memory dimension $d_{\mathrm m}$, then
$
\operatorname{rank}(M_{RS})\le d_{\mathrm m}^{2}.
 $
Hence
\begin{equation}
d_{\mathrm m}
\ge
\left\lceil
\sqrt{\operatorname{rank}(M_{RS})}
\right\rceil .
\label{eq:memory_dimension_lower_bound}
\end{equation}
\end{theorem}

The bound in Eq.~\eqref{eq:memory_dimension_lower_bound} is derived using tensor-network methods~\cite{wood2015tensor} in Appendix~\ref{app:tensor_network_memory_bound}, with an equivalent quantum-comb derivation given in Appendix~\ref{sec:comb-memory-rank}, and graphical conventions summarised in Appendix~\ref{app:preliminaries}.
The bound is one-sided: a higher rank of $M_{RS}$ certifies a larger memory dimension, but a low rank need not mean that the process has little memory.
The rank that can be observed also depends on the chosen reference-system input.
For a pure input of Schmidt rank $s=\operatorname{rank}(\rho_S)$, we have
\[
\operatorname{rank}(M_{RS})\le s\,d_S,
\]
where $d_S=\dim\mathcal H_S$.
Increasing $s$ from $1$ to $d_S$ raises this rank limit from $d_S$ to $d_S^2$ without changing the memory of the process.
A larger reference can therefore support a stronger memory-dimension certificate, while requiring characterisation on a larger retained $RS$ output.
As we shall see in the later section, we could choose a maximally entangled state to enhance the memory detection capability.

Operationally, $M_{RS}$ can be reconstructed directly from the retained circuit output, as detailed in Appendix~\ref{app:operational_reconstruction}. Control measurements in the $X$ and $Y$ bases, followed by conditional tomography on $RS$, determine the unnormalised conditional operators $\Omega_\alpha=p_\alpha\rho_{RS}^{\alpha}$, for $\alpha\in\{+,-,+i,-i\}$. 
Here $p_\alpha$ is the outcome probability and $\rho_{RS}^{\alpha}$ is the corresponding normalised conditional state. 
Their linear combination yields
\[
M_{RS}
=
\frac{1}{2}
\left[
\Omega_+-\Omega_-
-i\left(\Omega_{+i}-\Omega_{-i}\right)
\right].
\]
Since $M_{RS}$ need not be Hermitian, we determine its rank from its singular-value spectrum.

To summarise, the protocol has two essential parts: the two system-environment interactions, represented by $\mathcal C_1$ and $\mathcal C_2$, and the evolution of the environment between them.
The two interactions encode information about the interval memory dimension in $M_{RS}$.
The environmental evolution between them generates and shapes the memory that we seek to detect.
In Appendix~\ref{app:extensions-and-robustness}, we show that the protocol detects classical as well as quantum memory, and examine imperfect switch-off of the system-environment interaction during the interval.
Specifically, when the system-environment interaction cannot be fully switched off during the interval, the resolved rank $r_\epsilon$ can still provide a lower bound on the memory dimension, $d_{\mathrm m}\ge\lceil\sqrt{r_\epsilon}\rceil$, where
\[
r_\epsilon:=\#\{k:s_k(M_{RS})>\epsilon\},
\]
and the choice of $\epsilon$ depends on the form of the residual interaction.

\textit{Numerical protocol and setup.---}
The preceding analysis describes the memory connecting the two interaction processes in a simplified model without specifying the microscopic environmental dynamics.
We now apply the protocol to explicit many-body environments, where the environment $E$ evolves throughout the process and may also dissipate information.
During each of the two interactions, the joint state evolves for an interaction time $\tau$ under
\begin{equation}
\begin{aligned}
\mathcal L_{SE}(\rho_{SE})
={}&-i\left[
H_S+H_E+\lambda H_{SE},\rho_{SE}
\right] \\
&+\left(\mathrm{id}_S\otimes\mathcal D\right)(\rho_{SE}),
\end{aligned}
\label{eq:main_interaction_generator}
\end{equation}
whereas between the two interactions the system-environment coupling is switched off, $S$ remains idle, and $E$ evolves for an interval time $T$ under
\begin{equation}
\mathcal L_E(\rho_E)
=
-i\left[
H_E,\rho_E
\right]
+
\mathcal D(\rho_E).
\label{eq:main_interval_generator}
\end{equation}
Here $H_S$ and $H_E$ describe the internal dynamics of the system and environment, $H_{SE}$ is the system-environment interaction Hamiltonian, $\lambda$ is the coupling strength, and $\mathcal D$ is a dissipator acting on $E$. We study how the resolved rank depends on the interaction time $\tau$, the interval time $T$, the coupling strength $\lambda$, the internal dynamics of the environment, and the geometry of the system-environment contact.

In the dissipative simulations, the same dissipator remains active during both interactions and throughout the interval.
Unless stated otherwise, we use a boundary contact between $S_n$ and $E_1$, i.e., the last qubit in the system and the first qubit in the environment
\begin{equation}
H_{SE}^{\rm bdry}
=
\sum_{\mu\in\{x,y,z\}}
J_c^\mu\,
\sigma_{S,n}^\mu
\sigma_{E,1}^\mu,
\label{eq:main_boundary_contact}
\end{equation}
so only these two boundary sites are directly coupled.
The setup is illustrated in Fig.~\ref{fig:numerical_tau_wait}(a).

\begin{figure}[t]
    \centering
    \includegraphics[width=\linewidth]{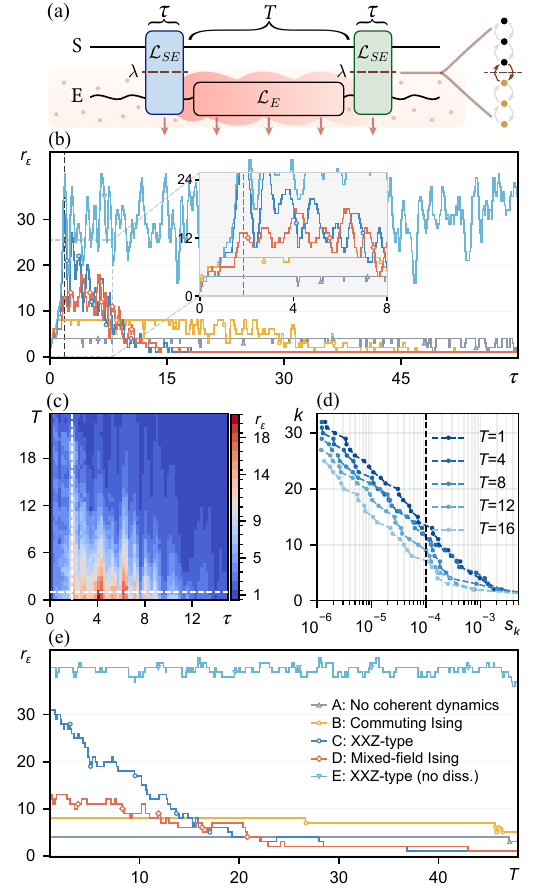}
    \caption{Temporal dependence of the resolved rank $r_\epsilon$.
    (a) Schematic of the numerical model: two system-environment interactions of duration $\tau$ are separated by an interval of duration $T$, with the boundary contact between $S_n$ and $E_1$.
    (b) $r_\epsilon$ versus interaction time $\tau$ for models A-E at $\lambda=1$ and $T=1$; the inset shows the short-time regime, and the vertical dashed line marks the reference time $\tau_0$.
    (c) $r_\epsilon(\tau,T)$ for model D at $\lambda=1$; the white dashed lines indicate the cuts shown in panels (b) and (e).
    (d) Singular values of $M_{RS}$ for model D at selected interval times; the vertical dashed line marks the threshold $\epsilon=10^{-4}$.
    (e) $r_\epsilon$ versus interval time $T$ at $\lambda=1$ and $\tau=\tau_0$.
    }
    \label{fig:numerical_tau_wait}
\end{figure}

To test the protocol across qualitatively different environmental dynamics, we use five representative cases that provide controlled comparisons of internal dynamics and dissipation.
Model A has $H_E=0$ and provides a no-propagation baseline.
Model B is a commuting longitudinal-field Ising chain.
Model C is an XXZ-type chain, widely used to study interacting one-dimensional spin dynamics
\cite{Bertini2021FiniteTemperature,Gopalakrishnan2023AnomalousTransport}.
Model D is an Ising chain with transverse and longitudinal fields, a standard nonintegrable spin-chain model
\cite{Kim2013Ballistic,Noh2021OperatorGrowth}.
Models A-D include local depolarisation acting on the environment, allowing us to examine how dissipation affects the memory resolved by the protocol.
Model E has the same coherent Hamiltonian as model C but no depolarisation, so the comparison between C and E isolates the effect of dissipation.
The full Hamiltonians, dissipator, and numerical parameters are given in Appendix~\ref{sec:numerical_investigation}.

For simulations in the following sections, we use the maximally entangled reference-system input
$ 
\ket{\Phi_{d_S}}_{RS}
=
\frac{1}{\sqrt{d_S}}
\sum_{j=0}^{d_S-1}
\ket{j}_R\ket{j}_S,
\label{eq:main_numerical_inputs}
$
which has full Schmidt rank $s=d_S$ {to maximise the memory detection capability as shown in Theorem \ref{thm:memory_bound}}.
The two auxiliary inputs and the initial environment state are maximally mixed,
$
\rho_1=\rho_2=\frac{I_{d_S}}{d_S} $,
 $
\rho_E(0)=\frac{I_{d_E}}{d_E}.
 $

\textit{Memory visibility with finite interaction time.---}
The interaction time $\tau$ sets the duration of each system-environment interaction.
At $\tau=0$, $S$ does not interact with $E$, so no information is exchanged and $M_{RS}$ remains at its rank-one baseline.
As $\tau$ increases, the environment starts to encode the memory information in the matrix $M_{RS}$.
Accordingly, $r_\epsilon$ initially increases for all five models, as shown in Fig.~\ref{fig:numerical_tau_wait}(b).
At longer interaction times, $r_\epsilon$ generally decreases for models A-D, whereas the nondissipative model E retains a larger rank.
Since models C and E have the same coherent Hamiltonian, their comparison shows that depolarisation contributes to this suppression.
Increasing $\tau$ therefore has two competing effects in the dissipative models: it allows more information of the memory dimension to be encoded in the matrix $M_{RS}$, but also gives depolarisation more time to act. 
We choose the reference time $\tau_0$ from the early high-rank region of the interaction-time scan and keep it fixed for all subsequent simulations.

\textit{Memory decay between the interactions.---}
We now show how the protocol resolves memory decay as the environment evolves between the two interactions, without requiring direct access to $E$. {Both $S$ and $E$ are taken to be three-qubit chains.}
We fix $\tau=\tau_0$ and $\lambda=1$ and vary the interval time $T$.
When depolarisation is present, increasing $T$ gives the dissipative dynamics in models A-D more time to weaken the information left by the first interaction. 
Accordingly, $r_\epsilon$ generally decreases with $T$ [Fig.~\ref{fig:numerical_tau_wait}(e)] as individual singular values fall below the threshold $\epsilon$ [Fig.~\ref{fig:numerical_tau_wait}(d)].
The comparison between models C and E, which differ only by depolarisation, further shows that dissipation contributes to this suppression.

The broader \((\tau,T)\) landscape in Fig.~\ref{fig:numerical_tau_wait}(c) shows that this overall decrease is accompanied by local increases. 
During the interval, $S$ is idle while $E$ continues to evolve.
Coherent dynamics can therefore redistribute information within $E$ and allow it to influence $S$ again before dissipation suppresses it.
Thus, the interval evolution can both redistribute environmental memory and reduce what remains visible in the second interaction.

\textit{Memory in many-body dynamics.---}
We now examine how many-body dynamics and coupling geometry affect the memory visible to the protocol.
We first use the boundary contact in Eq.~\eqref{eq:main_boundary_contact}, fix $\tau=\tau_0$ and $T=1$, and vary the coupling strength $\lambda$.
At $\lambda=0$, the system does not interact with the environment, so the environment cannot carry any influence from the first interaction to the second and hence $r_\epsilon=1$.
At weak and moderate coupling, increasing $\lambda$ makes this influence easier to detect, and the rank increases [Fig.~\ref{fig:numerical_coupling_geometry}(a)].

\begin{figure}[b]
    \centering
    \includegraphics[width=\linewidth]
    {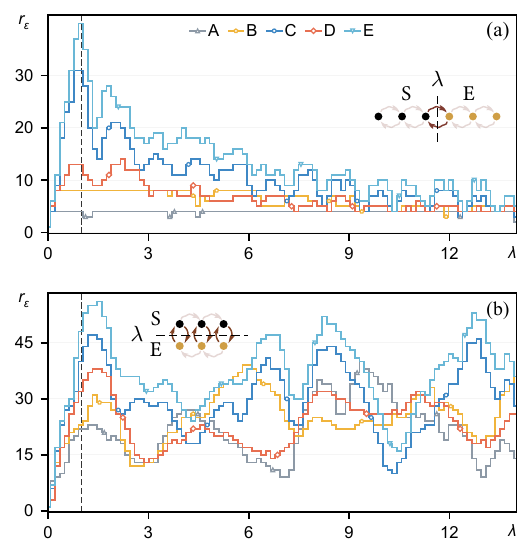}
    \caption{Rank $r_\epsilon$ versus coupling strength for two coupling geometries. Models A-E are evaluated at $\tau=\tau_0$, $T=1$, and $\epsilon=10^{-4}$.
    (a) Boundary contact coupling between $S_n$ and $E_1$, defined in Eq.~\eqref{eq:main_boundary_contact}.
    (b) Ladder contact coupling each $S_j$ directly to the corresponding $E_j$, defined in Eq.~\eqref{eq:main_ladder_contact}.
    The vertical dashed line marks $\lambda=1$ used in Fig.~\ref{fig:numerical_tau_wait}.
    }
    \label{fig:numerical_coupling_geometry}
\end{figure}

The rank initially increases in all cases, but how far it rises depends on the internal dynamics of the environment.
When $H_E=0$ (model A), the environment qubits do not interact, so only the directly coupled qubit $E_1$ can carry memory between the two interactions.
The rank then reaches the one-qubit ceiling $r_\epsilon=4$, giving the tight lower bound $d_{\mathrm m}\ge2$.
When the environment qubits interact, more than one qubit can contribute.
In the commuting Ising case (model B), the $E_1$-$E_2$ interaction allows $E_2$ to contribute, although the commuting dynamics still limits how far information can spread.
The rank reaches $r_\epsilon=8$, giving $d_{\mathrm m}\ge3$.
The remaining models allow broader spreading through the environment and reach still larger ranks.

At stronger coupling, however, the increase does not continue.
For the boundary contact, $r_\epsilon$ decreases again as $\lambda$ becomes large.
The same decrease appears in model E without depolarisation, showing that the downturn cannot be attributed to dissipation.
Since the boundary contact acts only on the $S_n$-$E_1$ link, a natural explanation is that this pair becomes so strongly coupled that the rest of the environment is harder to reach during the fixed interaction time.
To test this explanation, we replace the boundary contact in Eq.~\eqref{eq:main_boundary_contact} with a ladder contact that couples each system site to the corresponding environment site
\begin{equation}
H_{SE}^{\rm ladder}
=
\frac{1}{n}
\sum_{j=1}^{n}
\sum_{\mu\in\{x,y,z\}}
J_c^\mu
\sigma_{S,j}^\mu
\sigma_{E,j}^\mu .
\label{eq:main_ladder_contact}
\end{equation}
The factor $1/n$ gives the ladder and boundary contacts the same spectral width.
With the ladder contact, the ranks no longer show the same decrease at large $\lambda$, as shown in Fig.~\ref{fig:numerical_coupling_geometry}(b).
Directly coupling every environment site also lifts the boundary-contact rank ceilings found in the nonpropagating (A) and commuting-Ising (B) cases.
These results support the explanation that strong boundary coupling creates an access bottleneck to the rest of the environment, which is removed when all sites are coupled directly. 
The corresponding $(\lambda,T)$ landscapes are shown in Appendix~\ref{sec:numerical_investigation}.
Together, these results show that the memory visible to the protocol depends on both how information spreads within the environment and how the system accesses it. 

\textit{Discussion.---}
We have introduced a pseudo-control protocol that certifies a lower bound on the memory dimension of an open quantum process from accessible outputs.
This requires neither direct access to the environment nor reconstruction of the full multitime process.
Our numerical results show that the detectable memory depends on both information spreading within the environment and how the system couples to it, illustrating the protocol's ability to probe the structure of system-environment interactions.
Future work could explore simpler circuit implementations to reduce the experimental overhead.
The protocol could also help clarify how information propagation in many-body environments affects the memory dimension required to describe the resulting open-system dynamics.

\textit{Acknowledgements.---}
We thank Mile Gu, Yunlong Xiao, Zihan Hao, Ingo Roth, and Aditya Iyer for insightful discussions. As this work was nearing completion, we became aware of related forthcoming work by Lisa Kolesnyk and Ingo Roth~\cite{KolesnykRothForthcoming} on probing the effective environment dimension of non-Markovian, gate-dependent noise models through randomized benchmarking. 
C.L. acknowledges support from Quantum Science and Technology–National Science and Technology Major Project under Grant No. 2024ZD0300500. X.W. would like to acknowledge the support from the National Research Foundation through the Singapore Ministry of Education Tier 1 Grant RT4/23 and RG91/25 (S), the NRF Investigatorship on Quantum-Enhanced Agents (Grant No. NRF-NRFI09-0010) and the National Quantum Office, hosted in A*STAR, under its Centre for Quantum Technologies Funding Initiative (S24Q2d0009). A.Y. is supported by UKRI Future Leaders Fellowship (Grant No. 10128920).

\bibliography{references_v3}

\clearpage
\newpage
\widetext

\appendix

\let\addcontentsline\oldaddcontentsline


\addcontentsline{toc}{section}{Appendix}

\tableofcontents

\section{Preliminaries}
\label{app:preliminaries}

The main proofs are expressed in tensor-network (TN) notation. We begin by specifying the graphical conventions used throughout this appendix. This notation makes the contraction and factorisation structure of the relevant tensors explicit.

\subsection{Graphical dictionary}

A tensor is represented by a node with attached lines, or legs. Each leg is associated with a Hilbert space and has the dimension of that space. An open leg represents an uncontracted (or free) index. More generally, a node with $n$ open legs represents an $n$-index tensor. The elementary conventions are summarized in Fig.~\ref{fig:a1_tensors}. 

\begin{figure}[b]
    \centering
    \includegraphics[width=0.82\textwidth]{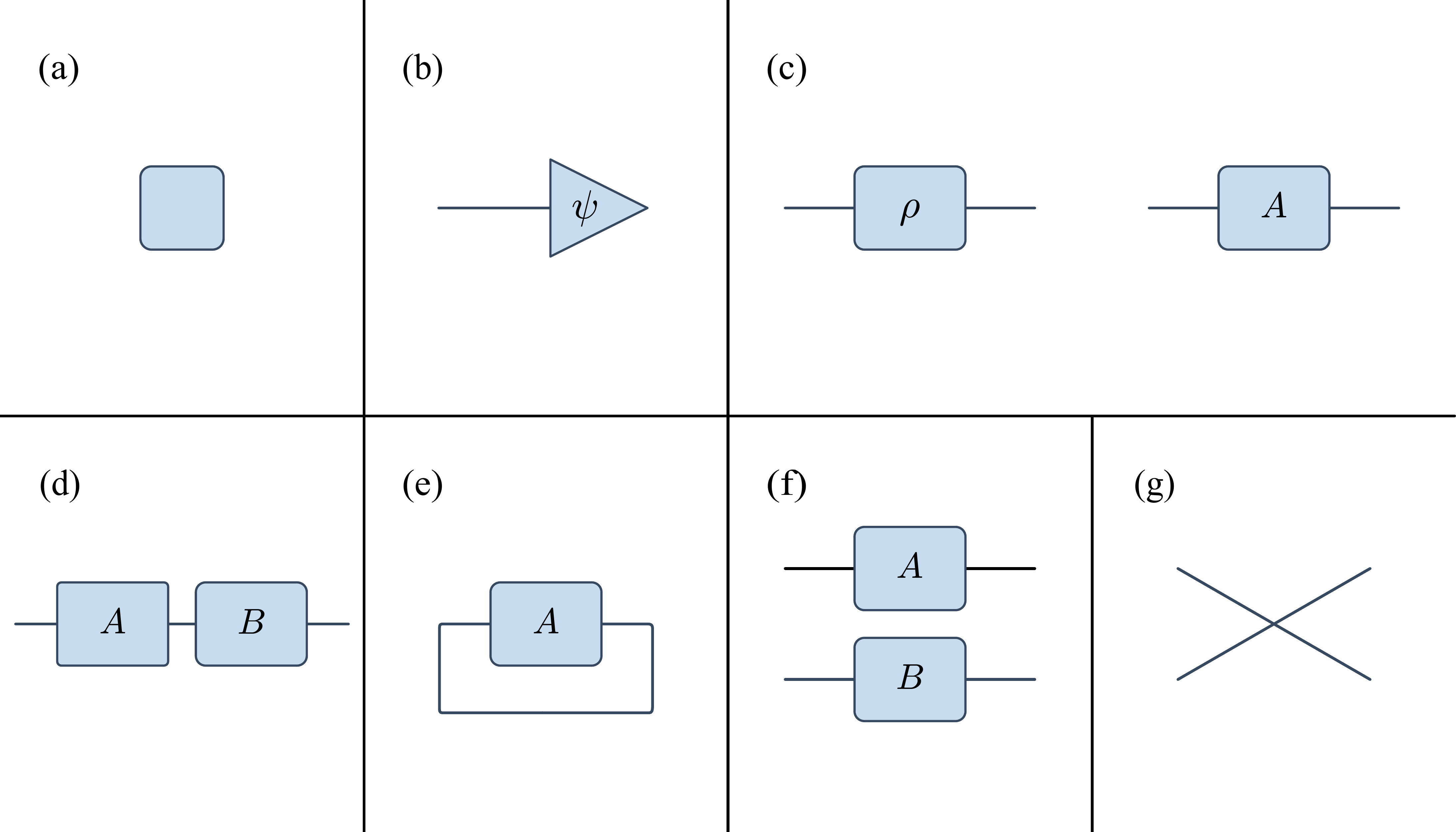}
    \caption{Graphical TN conventions.
        (a) A scalar.
        (b) A vector.
        (c) An operator.
        (d) Index contraction, illustrated by matrix multiplication.
        (e) The trace of an operator.
        (f) A tensor product.
        (g) A permutation of tensor factors, illustrated by the \texttt{SWAP} map.
    }
    \label{fig:a1_tensors}
\end{figure}

A tensor with no open legs represents a scalar, as shown in Fig.~\ref{fig:a1_tensors}(a). A vector
\begin{equation}
    \ket{\psi}=\sum_i \psi_i \ket{i}
\end{equation}
is represented by a node with one open leg, as in Fig.~\ref{fig:a1_tensors}(b). An operator
\begin{equation}
    A=\sum_{i,j} A_{ij}\ket{i}\!\bra{j}
\end{equation}
is represented by a node with two open legs, corresponding to its ket and bra indices shown in Fig.~\ref{fig:a1_tensors}(c). Density operators are represented in the same way. Connecting two compatible legs denotes contraction of the corresponding indices. For example, joining the legs associated with the common index $j$ of $A_{ij}$ and $B_{jk}$ gives
\begin{equation}
    (AB)_{ik}=\sum_j A_{ij}B_{jk},
\end{equation}
as illustrated in Fig.~\ref{fig:a1_tensors}(d). Connecting the ket and bra legs of an operator denotes its trace
\begin{equation}
    \operatorname{Tr}(A)=\sum_i A_{ii},
\end{equation}
as shown in Fig.~\ref{fig:a1_tensors}(e). The juxtaposition of disconnected tensors denotes their tensor product, as in Fig.~\ref{fig:a1_tensors}(f). A crossing at which no tensor is placed denotes only a permutation of tensor factors. In particular, the crossing shown in Fig.~\ref{fig:a1_tensors}(g) represents the \texttt{SWAP} map. Such a crossing does not introduce an additional tensor or physical operation beyond the indicated permutation.

Throughout the proofs, we use two visually distinct types of boxes. Sharp-cornered boxes represent physical circuit elements, whereas rounded boxes represent tensors obtained after contractions, regrouping of indices, or coarse-graining.

\section{Operational reconstruction of the interference matrix}
\label{app:operational_reconstruction}

The interference matrix $M_{RS}$ is not itself a directly measured observable. Here we show how it can be reconstructed from experimentally accessible outputs of the pseudo-control circuit. Measurements of the control in the $X$ and $Y$ bases, combined with conditional tomography of $RS$, determine the two components needed to recover $M_{RS}$ exactly.

After the pseudo-control circuit and tracing out all other registers, let $\omega_{cRS}$ denote the retained state on $cRS$. In the computational basis of the control, we write
\begin{equation}
\omega_{cRS}
=
\sum_{i,j=0}^{1}
\ket{i}\!\bra{j}_{c}
\otimes
\omega_{ij}^{RS}.
\label{eq:appendix_block_decomposition}
\end{equation}
The interference matrix is the off-diagonal control block
\begin{equation}
M_{RS}
:=
\omega_{01}^{RS}.
\label{eq:appendix_M_definition}
\end{equation}
Since $\omega_{cRS}$ is Hermitian
\begin{equation}
\omega_{10}^{RS}
=
M_{RS}^{\dagger}.
\label{eq:appendix_M_dagger}
\end{equation}
The operator $M_{RS}$ is generally neither Hermitian nor positive. As an off-diagonal control block, it is not obtained from a computational-basis measurement of the control alone. Instead, it can be reconstructed from measurements of the control in the $X$ and $Y$ bases together with conditional tomography on the retained $RS$ output.

For a control outcome
$\alpha\in\{+,-,+i,-i\}$,
let $p_{\alpha}$ denote its probability and
$\rho_{RS}^{\alpha}$ the corresponding normalised conditional state of $RS$.
We define the associated unnormalised conditional operator
\begin{equation}
\Omega_{\alpha}
:=
p_{\alpha}\rho_{RS}^{\alpha}
=
\bra{\alpha}_{c}
\omega_{cRS}
\ket{\alpha}_{c}.
\label{eq:appendix_conditional_operator}
\end{equation}
The probability $p_{\alpha}$ is obtained from the control measurement, while
$\rho_{RS}^{\alpha}$ is obtained by tomography conditioned on that outcome.

\paragraph{$X$-basis readout.}

We first measure the control in the $X$ basis
\begin{equation}
\ket{\pm}_{c}
=
\frac{1}{\sqrt{2}}
\left(
\ket{0}_{c}
\pm
\ket{1}_{c}
\right).
\label{eq:appendix_X_basis}
\end{equation}
Experimentally, this measurement can be implemented by applying a Hadamard gate to the control followed by computational-basis readout. Substituting Eq.~\eqref{eq:appendix_block_decomposition} into Eq.~\eqref{eq:appendix_conditional_operator} gives
\begin{align}
\Omega_{+}
&=
\frac{1}{2}
\left(
\omega_{00}^{RS}
+
M_{RS}
+
M_{RS}^{\dagger}
+
\omega_{11}^{RS}
\right),
\label{eq:appendix_plus_expansion}
\\
\Omega_{-}
&=
\frac{1}{2}
\left(
\omega_{00}^{RS}
-
M_{RS}
-
M_{RS}^{\dagger}
+
\omega_{11}^{RS}
\right).
\label{eq:appendix_minus_expansion}
\end{align}
Taking the difference cancels the two diagonal control blocks
\begin{equation}
\Omega_{+}
-
\Omega_{-}
=
M_{RS}
+
M_{RS}^{\dagger}.
\label{eq:appendix_X_difference}
\end{equation}
Thus, the $X$-basis data determine the combination
$M_{RS}+M_{RS}^{\dagger}$.

\paragraph{$Y$-basis readout.}

We next measure the control in the $Y$ basis
\begin{equation}
\ket{\pm i}_{c}
=
\frac{1}{\sqrt{2}}
\left(
\ket{0}_{c}
\pm
i\ket{1}_{c}
\right).
\label{eq:appendix_Y_basis}
\end{equation}
This measurement can be implemented by applying $S^{\dagger}$ followed by a Hadamard gate to the control before computational-basis readout. Using Eq.~\eqref{eq:appendix_conditional_operator}, we obtain
\begin{align}
\Omega_{+i}
&=
\frac{1}{2}
\left(
\omega_{00}^{RS}
+
iM_{RS}
-
iM_{RS}^{\dagger}
+
\omega_{11}^{RS}
\right),
\label{eq:appendix_plusi_expansion}
\\
\Omega_{-i}
&=
\frac{1}{2}
\left(
\omega_{00}^{RS}
-
iM_{RS}
+
iM_{RS}^{\dagger}
+
\omega_{11}^{RS}
\right).
\label{eq:appendix_minusi_expansion}
\end{align}
Their difference is
\begin{equation}
\Omega_{+i}
-
\Omega_{-i}
=
i
\left(
M_{RS}
-
M_{RS}^{\dagger}
\right).
\label{eq:appendix_Y_difference}
\end{equation}

\paragraph{Reconstruction of $M_{RS}$.}

Combining Eqs.~\eqref{eq:appendix_X_difference} and
\eqref{eq:appendix_Y_difference} gives
\begin{equation}
M_{RS}
=
\frac{1}{2}
\left[
\Omega_{+}
-
\Omega_{-}
-
i
\left(
\Omega_{+i}
-
\Omega_{-i}
\right)
\right].
\label{eq:appendix_M_reconstruction}
\end{equation}
Hence the four unnormalised conditional operators obtained from the $X$- and $Y$-basis measurements reconstruct the interference matrix exactly. Equivalently
\begin{equation}
M_{RS}
=
\frac{1}{2}
\left[
p_{+}\rho_{RS}^{+}
-
p_{-}\rho_{RS}^{-}
-
i
\left(
p_{+i}\rho_{RS}^{+i}
-
p_{-i}\rho_{RS}^{-i}
\right)
\right].
\label{eq:appendix_M_reconstruction_explicit}
\end{equation}
Because $M_{RS}$ is generally non-Hermitian, its rank is determined from its singular-value spectrum rather than its eigenvalue spectrum.

\section{Tensor-network proof of Theorem 1}
\label{app:tensor_network_memory_bound}

We now prove Theorem~\ref{thm:memory_bound} using the tensor-network representation of the pseudo-control circuit. The key observation is that the rank of the interference matrix is limited by the dimensions of the tensor-network bonds crossing an appropriate matrix-rank cut. A memory of dimension $d_{\mathrm m}$ contributes one such bond on the forward branch and one on the backward branch, so the two memory indices together provide at most $d_{\mathrm m}^{2}$ independent contributions.

For clarity, we first suppress the reference $R$ and work with the system-only interference matrix
\begin{equation}
M:=\omega_{01}^{S}.
\label{eq:appendix_system_only_M}
\end{equation}
This setting makes the tensor structure and the rank-cut argument easiest to see. In this section, we first derive the six circuit steps shared by memoryless and finite-memory realisations, and then compare their contraction patterns. These results prove the memory bound in Theorem \ref{thm:memory_bound} (for a simplified setting in Fig.~\ref{fig:pseudocontrol} with $R = I$ and unitary channels).
At the end, we restore the isolated reference $R$ and show that the same memory bound applies to the reference-assisted interference matrix $M_{RS}$ used in the main text.

\subsection{Circuit notation and the off-diagonal input sector}
\label{subsec:appendix_C_setup}

To keep the notation compact, we label the five non-control registers as
\begin{equation}
(1,2,3,4,5)
=
(\rho_{1},\sigma_{1},\rho,\sigma_{2},\rho_{2}),
\label{eq:appendix_register_order}
\end{equation}
where $\rho\equiv\rho_{S}$ is the system input. The states $\rho_{1}$ and $\rho_{2}$ are the two auxiliary system inputs, while $\sigma_{1}$ and $\sigma_{2}$ are the local environmental inputs to the first and second interactions. Register $3$ is the routed system register and contains the retained system output after the full circuit.

\begin{figure}[htbp]
    \centering
    \includegraphics[width=0.80\linewidth]{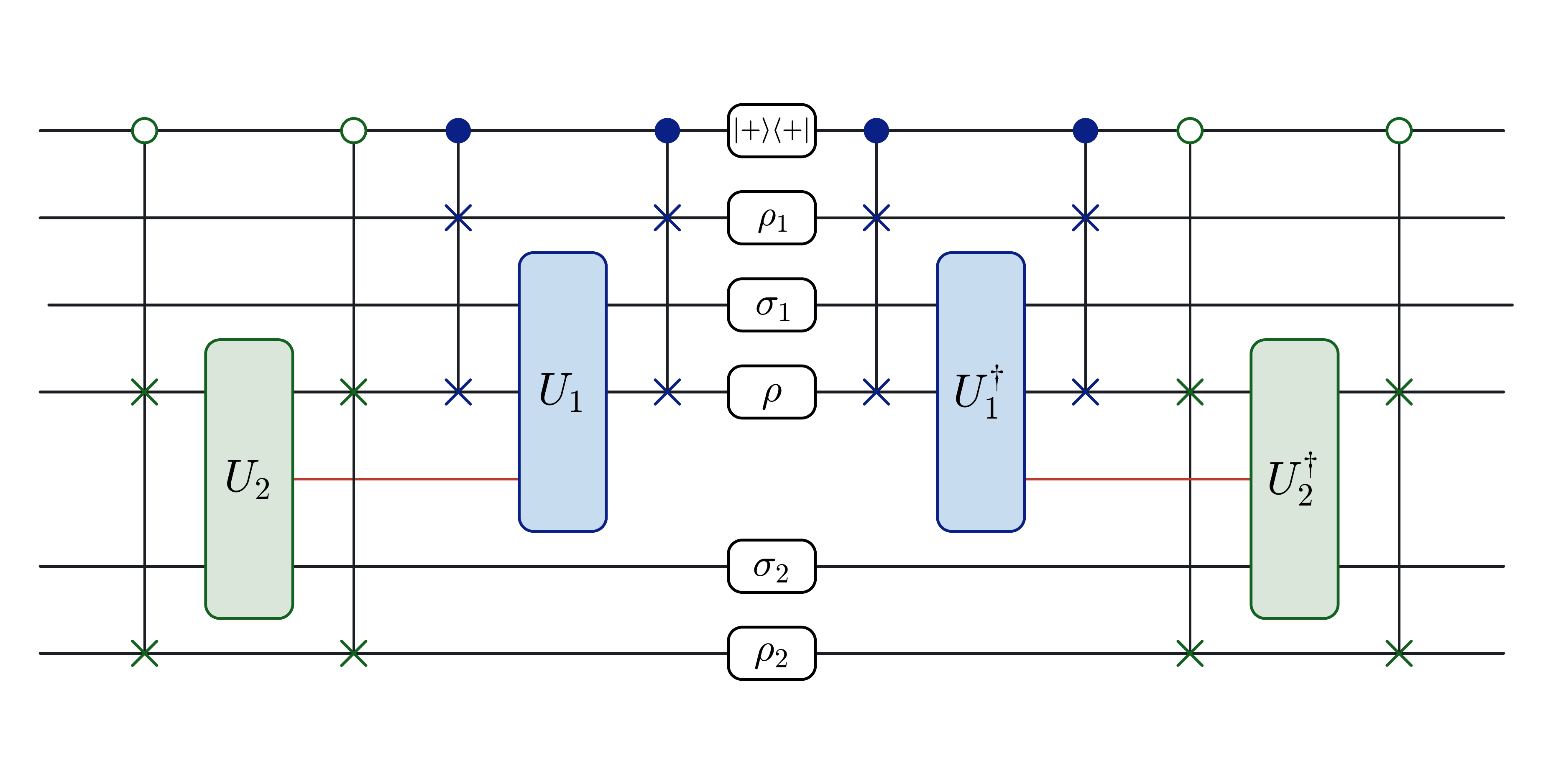}
    \caption{TN representation of the pseudo-control circuit used in the proof. The five non-control registers are ordered as in Eq.~\eqref{eq:appendix_register_order}.
    }
    \label{fig:C1_circuit_to_tensor}
\end{figure}
The initial operator on these five registers is
\begin{equation}
\eta
:=
\rho_{1}
\otimes
\sigma_{1}
\otimes
\rho
\otimes
\sigma_{2}
\otimes
\rho_{2}.
\label{eq:appendix_eta_definition}
\end{equation}
The control is initialised in
\begin{equation}
\ket{+}\!\bra{+}_{c}
=
\frac{1}{2}
\left(
\ket{0}\!\bra{0}_{c}
+
\ket{0}\!\bra{1}_{c}
+
\ket{1}\!\bra{0}_{c}
+
\ket{1}\!\bra{1}_{c}
\right).
\label{eq:appendix_control_plus_expansion}
\end{equation}
All controlled operations in the protocol are diagonal in the computational basis of the control, while the two interaction unitaries act trivially on $c$. The four control sectors therefore propagate independently. Since the interference matrix is the coefficient of the final $\ket{0}\!\bra{1}_{c}$ sector, it is sufficient to follow only the corresponding component of the initial operator. We define
\begin{equation}
\omega^{(0)}
:=
\frac{1}{2}
\ket{0}\!\bra{1}_{c}
\otimes
\eta.
\label{eq:appendix_omega0_definition}
\end{equation}
Here and below, $\omega^{(k)}$ denotes the propagated
$\ket{0}\!\bra{1}_{c}$ contribution after the $k$-th routing step. These operators are not density operators. Parenthesised superscripts label routing steps, whereas subscripts such as those in $\omega_{ij}^{S}$ label control blocks of the final reduced state. We define the two \texttt{SWAP} operators
\begin{equation}
S_{13}
:=
\operatorname{SWAP}_{1,3},
\qquad
S_{35}
:=
\operatorname{SWAP}_{3,5},
\label{eq:appendix_swap_definitions}
\end{equation}
and the corresponding controlled-\texttt{SWAP} gates
\begin{align}
C_{13}^{(1)}
&:=
\ket{0}\!\bra{0}_{c}\otimes I
+
\ket{1}\!\bra{1}_{c}\otimes S_{13},
\label{eq:appendix_C13_definition}
\\
C_{35}^{(0)}
&:=
\ket{0}\!\bra{0}_{c}\otimes S_{35}
+
\ket{1}\!\bra{1}_{c}\otimes I.
\label{eq:appendix_C35_definition}
\end{align}
Thus, $C_{13}^{(1)}$ applies $S_{13}$ only on the control-$\ket{1}_{c}$ branch, whereas $C_{35}^{(0)}$ applies $S_{35}$ only on the control-$\ket{0}_{c}$ branch. For the two interaction unitaries, we write
\begin{equation}
\widetilde U_{1}
:=
I_{1}\otimes U_{1}\otimes I_{45},
\qquad
\widetilde U_{2}
:=
I_{12}\otimes U_{2}\otimes I_{5},
\label{eq:appendix_embedded_U_definitions}
\end{equation}
so that $U_{1}$ acts on registers $(2,3)$ and $U_{2}$ acts on registers $(3,4)$. With these definitions, the off-diagonal control sector can be propagated through the six routing steps directly.

\subsection{Step-by-step propagation of the off-diagonal control sector}
\label{subsec:appendix_six_steps}

We now propagate the off-diagonal component in
Eq.~\eqref{eq:appendix_omega0_definition} through the circuit. The six circuit steps are
\begin{equation}
C_{13}^{(1)}
\;\longrightarrow\;
\widetilde U_{1}
\;\longrightarrow\;
C_{13}^{(1)}
\;\longrightarrow\;
C_{35}^{(0)}
\;\longrightarrow\;
\widetilde U_{2}
\;\longrightarrow\;
C_{35}^{(0)}.
\label{eq:appendix_six_step_sequence}
\end{equation}
The first three steps route the system through the first interaction, and the last three do the same for the second interaction. We work out the first controlled-\texttt{SWAP} explicitly; the remaining controlled-\texttt{SWAP}s follow from the same algebra.

\begin{figure}[htbp]
    \centering
    \includegraphics[width=0.82\linewidth]{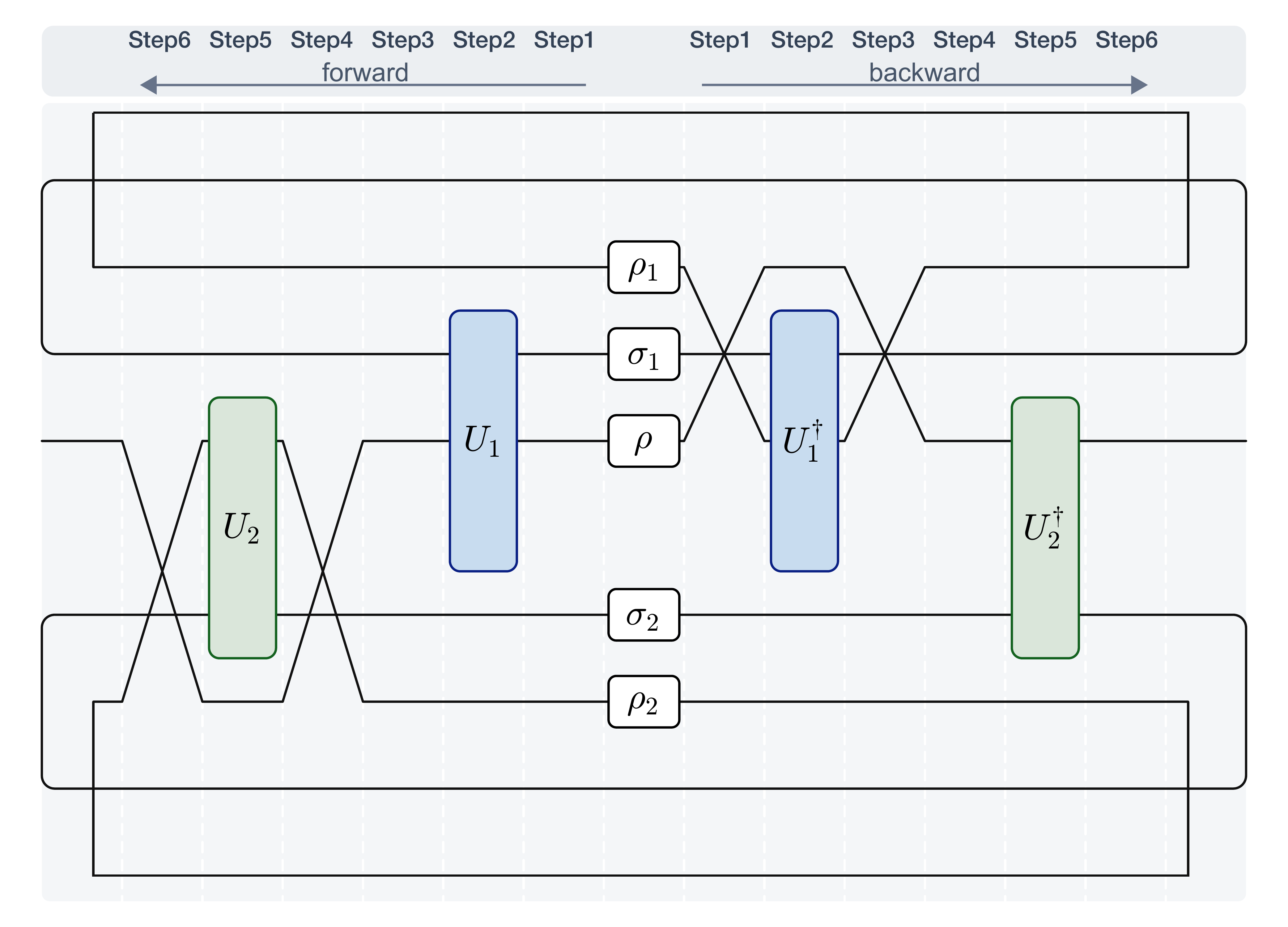}
    \caption{Step-by-step TN representation of the
    $\ket{0}\!\bra{1}_{c}$ contribution.
    The six steps correspond to
    Eqs.~\eqref{eq:appendix_step1_result}-\eqref{eq:appendix_step6_result}.
    The displayed circuit steps are the same for memoryless and finite-memory realisations; their difference enters only through the contraction of hidden indices between the two interactions.
    }
    \label{fig:C2_steps_no_mempry}
\end{figure}

\paragraph{Step 1: routing into the first interaction.}

We first apply the $\ket{1}_{c}$-controlled \texttt{SWAP}
$C_{13}^{(1)}$ between registers $1$ and $3$
\begin{equation}
\omega^{(1)}
=
C_{13}^{(1)}
\omega^{(0)}
C_{13}^{(1)\dagger}.
\label{eq:appendix_step1_start}
\end{equation}
Since $S_{13}=S_{13}^{\dagger}$, the controlled-\texttt{SWAP}
$C_{13}^{(1)}$ is also Hermitian. Substituting
Eqs.~\eqref{eq:appendix_omega0_definition} and
\eqref{eq:appendix_C13_definition} gives
\begin{align}
\omega^{(1)}
={}&
\frac{1}{2}
\Bigl(
\ket{0}\!\bra{0}_{c}\otimes I
+
\ket{1}\!\bra{1}_{c}\otimes S_{13}
\Bigr)
\Bigl(
\ket{0}\!\bra{1}_{c}\otimes\eta
\Bigr)
\nonumber\\
&\hspace{2.5cm}\times
\Bigl(
\ket{0}\!\bra{0}_{c}\otimes I
+
\ket{1}\!\bra{1}_{c}\otimes S_{13}
\Bigr).
\label{eq:appendix_step1_expansion}
\end{align}
The only nonzero control-projector products are
\begin{equation}
\ket{0}\!\bra{0}\,
\ket{0}\!\bra{1}
=
\ket{0}\!\bra{1},
\qquad
\ket{0}\!\bra{1}\,
\ket{1}\!\bra{1}
=
\ket{0}\!\bra{1}.
\label{eq:appendix_step1_projector_identities}
\end{equation}
All other combinations vanish. Therefore
\begin{equation}
\omega^{(1)}
=
\frac{1}{2}
\ket{0}\!\bra{1}_{c}
\otimes
\eta S_{13}.
\label{eq:appendix_step1_result}
\end{equation}
For the $\ket{0}\!\bra{1}_{c}$ sector, the left action corresponds to the control value $0$ on the ket branch, while the right action corresponds to the control value $1$ on the bra branch. We refer to these as the forward and backward branches, respectively. Thus, the first controlled-\texttt{SWAP} acts trivially on the forward branch and applies $S_{13}$ on the backward branch, where it exchanges the roles of $\rho_{1}$ and $\rho$. The remaining controlled-\texttt{SWAP}s are evaluated in the same way, so below we state their action directly.

\paragraph{Step 2: first interaction $U_{1}$.}

The first interaction acts on registers $(2,3)$ on both control branches. Hence
\begin{align}
\omega^{(2)}
&=
\widetilde U_{1}
\omega^{(1)}
\widetilde U_{1}^{\dagger}
\nonumber\\
&=
\frac{1}{2}
\ket{0}\!\bra{1}_{c}
\otimes
\widetilde U_{1}
\eta S_{13}
\widetilde U_{1}^{\dagger}.
\label{eq:appendix_step2_result}
\end{align}
On the forward branch, registers $(2,3)$ contain
$(\sigma_{1},\rho)$, so $U_{1}$ acts on the system input $\rho$.
On the backward branch, the preceding \texttt{SWAP} routes
$\rho_{1}$ into register $3$, so $U_{1}^{\dagger}$ acts on
$(\sigma_{1},\rho_{1})$.

\paragraph{Step 3: routing out of the first interaction.}

We apply $C_{13}^{(1)}$ a second time. As in Step~1, it acts trivially from the left and as $S_{13}$ from the right in the
$\ket{0}\!\bra{1}_{c}$ sector. Therefore
\begin{align}
\omega^{(3)}
&=
C_{13}^{(1)}
\omega^{(2)}
C_{13}^{(1)\dagger}
\nonumber\\
&=
\frac{1}{2}
\ket{0}\!\bra{1}_{c}
\otimes
\widetilde U_{1}
\eta S_{13}
\widetilde U_{1}^{\dagger}
S_{13}.
\label{eq:appendix_step3_result}
\end{align}
This second controlled-\texttt{SWAP} restores the visible register ordering after the first interaction while preserving the coherent distinction between the two routing histories.

\paragraph{Step 4: routing into the second interaction.}

We next apply the $\ket{0}_{c}$-controlled \texttt{SWAP}
$C_{35}^{(0)}$ between registers $3$ and $5$. In the
$\ket{0}\!\bra{1}_{c}$ sector, the control value is $0$ on the forward branch and $1$ on the backward branch. Thus $S_{35}$ acts from the left, while the backward branch is unchanged
\begin{align}
\omega^{(4)}
&=
C_{35}^{(0)}
\omega^{(3)}
C_{35}^{(0)\dagger}
\nonumber\\
&=
\frac{1}{2}
\ket{0}\!\bra{1}_{c}
\otimes
S_{35}
\widetilde U_{1}
\eta S_{13}
\widetilde U_{1}^{\dagger}
S_{13}.
\label{eq:appendix_step4_result}
\end{align}
Immediately before the second interaction, register $3$ therefore contains $\rho_{2}$ on the forward branch and $\rho$ on the backward branch.

\paragraph{Step 5: second interaction $U_{2}$.}

The second interaction acts on registers $(3,4)$ on both branches
\begin{align}
\omega^{(5)}
&=
\widetilde U_{2}
\omega^{(4)}
\widetilde U_{2}^{\dagger}
\nonumber\\
&=
\frac{1}{2}
\ket{0}\!\bra{1}_{c}
\otimes
\widetilde U_{2}
S_{35}
\widetilde U_{1}
\eta S_{13}
\widetilde U_{1}^{\dagger}
S_{13}
\widetilde U_{2}^{\dagger}.
\label{eq:appendix_step5_result}
\end{align}
On the forward branch, $U_{2}$ acts on
$(\rho_{2},\sigma_{2})$.
On the backward branch, $U_{2}^{\dagger}$ acts on
$(\rho,\sigma_{2})$.

\paragraph{Step 6: routing out of the second interaction.}

Finally, we apply $C_{35}^{(0)}$ once more. It again acts as
$S_{35}$ from the left and as the identity from the right, giving
\begin{align}
\omega^{(6)}
&=
C_{35}^{(0)}
\omega^{(5)}
C_{35}^{(0)\dagger}
\nonumber\\
&=
\frac{1}{2}
\ket{0}\!\bra{1}_{c}
\otimes
S_{35}
\widetilde U_{2}
S_{35}
\widetilde U_{1}
\eta S_{13}
\widetilde U_{1}^{\dagger}
S_{13}
\widetilde U_{2}^{\dagger}.
\label{eq:appendix_step6_result}
\end{align}
The final controlled-\texttt{SWAP} returns the retained system output to register $3$.

\paragraph{Final partial trace.}

The coefficient of $\ket{0}\!\bra{1}_{c}$ in
Eq.~\eqref{eq:appendix_step6_result} is the full off-diagonal operator before the auxiliary registers are discarded. Tracing out all registers except the retained system therefore gives
\begin{equation}
M
=
\frac{1}{2}
\operatorname{Tr}_{1,2,4,5}
\left[
S_{35}
\widetilde U_{2}
S_{35}
\widetilde U_{1}
\eta S_{13}
\widetilde U_{1}^{\dagger}
S_{13}
\widetilde U_{2}^{\dagger}
\right].
\label{eq:appendix_M_full_contraction}
\end{equation}
In the five-register notation used here,
$\operatorname{Tr}_{1,2,4,5}$ discards all displayed registers other than register $3$. More generally, the partial trace also includes any local environmental outputs that are not retained as part of the system output. The overall factor $1/2$ originates from the initial
$\ket{0}\!\bra{1}_{c}$ coherence of the control state
$\ket{+}\!\bra{+}_{c}$. This scalar factor does not affect any of the structural arguments below. Eq.~\eqref{eq:appendix_M_full_contraction} is the exact system-only interference matrix whose TN contraction pattern we analyse next.

\subsection{Memoryless contraction pattern}
\label{subsec:appendix_markovian_contraction}

We first consider the memoryless case, which provides the simplest example of the matrix-rank cut used below. In a memoryless realisation, no hidden environmental degree of freedom leaving the first interaction is passed to the second. Consequently, all hidden indices associated with each interaction are contracted locally, and no hidden bond connects the two interactions. After tracing out the discarded registers and rearranging the TN as shown in Figs.~\ref{fig:C3_tensor_no_memory} and \ref{fig:C4_tensor_simple}, the network is connected across the cut only through the system input $\rho$. We therefore group all tensors on the left and right sides of the cut into two contracted blocks, denoted by $A$ and $B$, respectively. This gives
\begin{equation}
M
=
\frac{1}{2}
A\rho B.
\label{eq:appendix_markovian_factorization}
\end{equation}
The labels $A$ and $B$ simply denote the complete contractions on the two sides of the cut. After the TN is rearranged, tensors from either interaction may be absorbed into either block. The essential point is that the only connection across the cut is the system input $\rho$: no hidden environmental bond crosses it.

For this system-only warm-up, we take the system input to be pure
\begin{equation}
\rho
=
\ket{\psi}\!\bra{\psi}.
\end{equation}
Eq.~\eqref{eq:appendix_markovian_factorization} then becomes
\begin{align}
M
&=
\frac{1}{2}
A
\ket{\psi}\!\bra{\psi}
B
\nonumber\\
&=
\frac{1}{2}
\bigl(
A\ket{\psi}
\bigr)
\bigl(
\bra{\psi}B
\bigr).
\label{eq:appendix_markovian_outer_product}
\end{align}
Thus $M$ is a single outer product. If both vectors in
Eq.~\eqref{eq:appendix_markovian_outer_product} are nonzero, this outer product has rank one; if either vector vanishes, then $M=0$ and its rank is zero. 
In either case
\begin{equation}
\operatorname{rank}(M)
\leq
1.
\label{eq:appendix_markovian_rank_bound}
\end{equation}

\begin{figure}[htbp]
    \centering
    \includegraphics[width=0.8\linewidth]{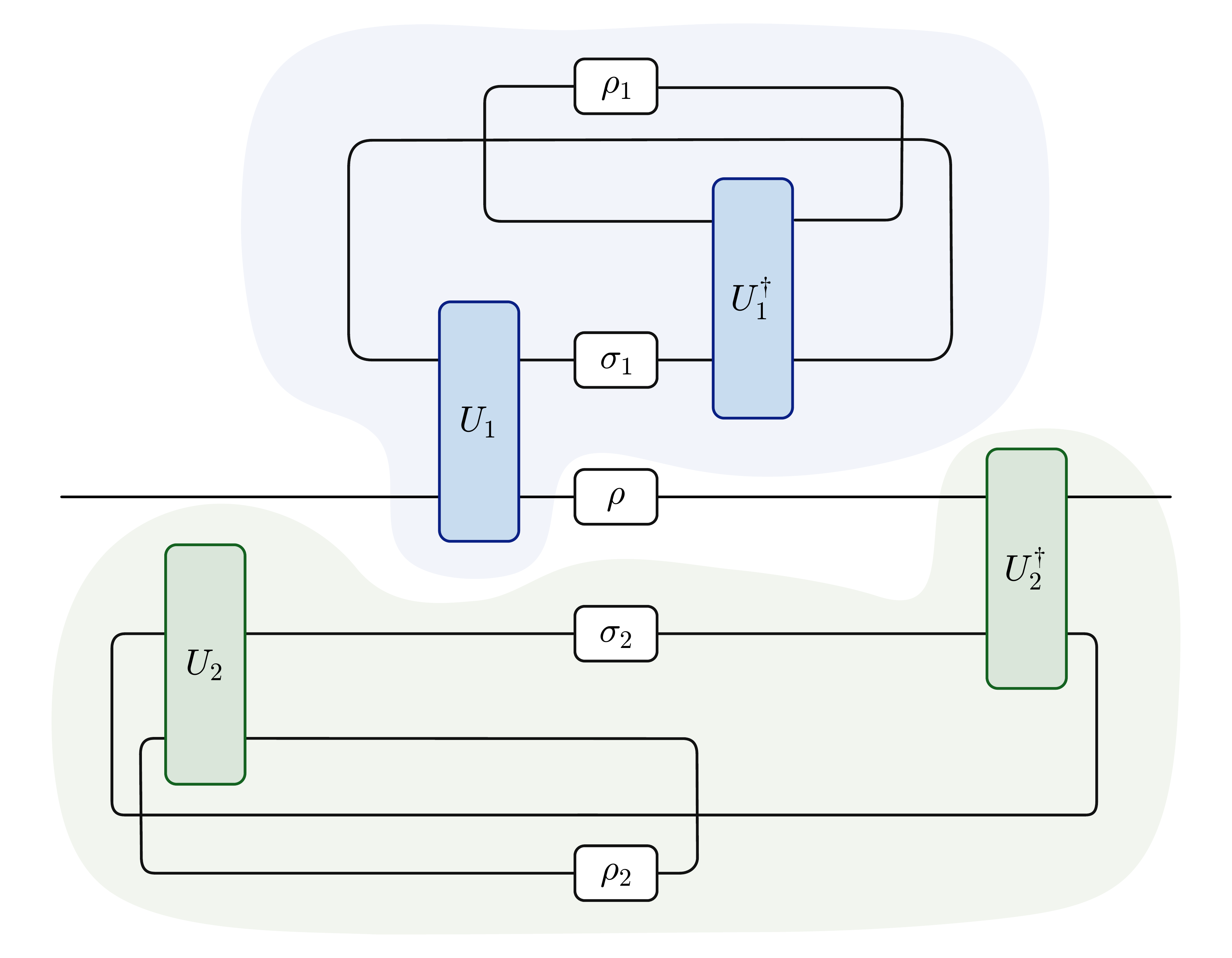}
    \caption{Memoryless contraction pattern after tracing out the discarded registers and rearranging the TN. No hidden bond connects the two interactions.
    }
    \label{fig:C3_tensor_no_memory}
\end{figure}

\begin{figure}[htbp]
    \centering
    \includegraphics[width=0.44\linewidth]{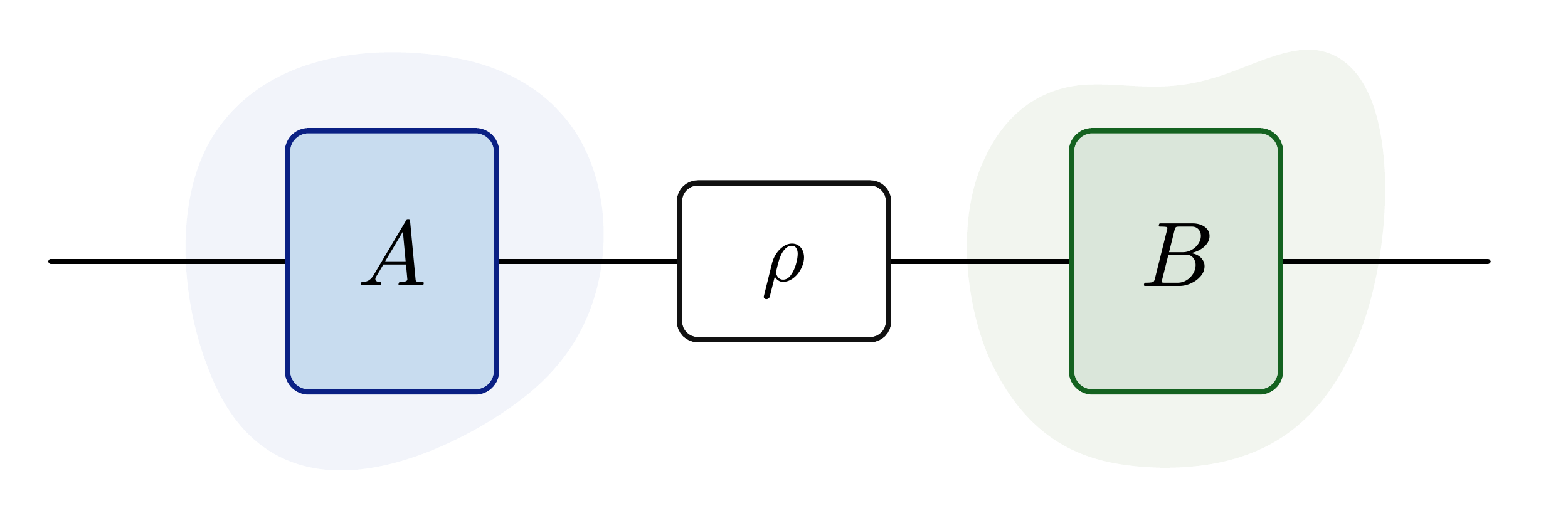}
    \caption{Simplified memoryless network. The two contracted blocks are connected only through the system input $\rho$. For a pure input, $\rho=\ket{\psi}\!\bra{\psi}$ factorises across the cut, so the resulting matrix $M$ is a single outer product.
    }
    \label{fig:C4_tensor_simple}
\end{figure}

The finite-memory case differs only in that additional memory bonds connect the two contracted blocks.

\subsection{Finite-memory contraction pattern and proof of the memory bound}
\label{subsec:appendix_finite_memory_contraction}

We now allow an effective memory carrier to connect the two interactions. Let $\mathcal H_{m}$ denote its Hilbert space, with
\begin{equation}
d_{\mathrm m}
:=
\dim\mathcal H_{m}.
\label{eq:appendix_dm_definition}
\end{equation}
At the circuit level, this is a single physical memory carrier passed from the first interaction to the second. In the operator TN, however, the ket and bra copies of this carrier appear as two separate bonds: one on the forward branch and one on the backward branch. Each bond has dimension $d_{\mathrm m}$.

The six visible circuit steps derived in Sec.~\ref{subsec:appendix_six_steps} are unchanged. Finite memory changes only the hidden contraction between the two interactions. Fig.~\ref{fig:C5_steps_with_mempry} shows the same circuit propagation with the two memory bonds made explicit.

\begin{figure}[htbp]
    \centering
    \includegraphics[width=0.82\linewidth]{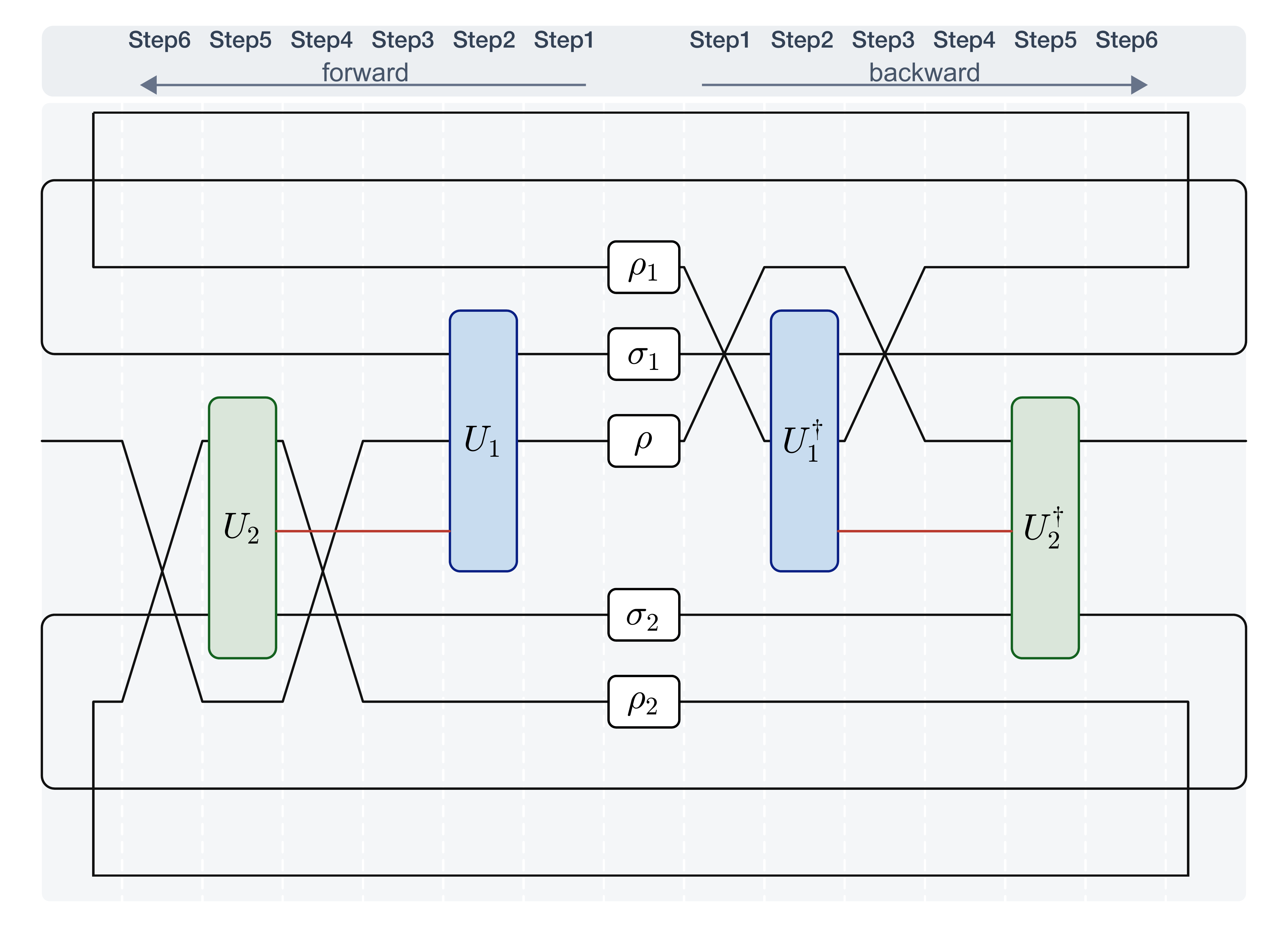}
    \caption{The same six circuit steps in a finite-memory realisation. The red wire represents the physical memory carrier in $\mathcal H_m$. In the operator TN, its ket and bra indices appear as separate forward and backward memory bonds, each of dimension $d_{\mathrm m}$.
    }
    \label{fig:C5_steps_with_mempry}
\end{figure}

After tracing out the discarded registers and rearranging the TN as in the memoryless case, the two contracted blocks are now connected not only through the system input $\rho$, but also through the forward and backward memory bonds. This structure is shown in Figs.~\ref{fig:C6_tensor_with_memory} and \ref{fig:C7_tensor_simple}. Relative to the memoryless case, these two memory bonds are the only additional connections across the matrix-rank cut.

\begin{figure}[htbp]
    \centering
    \includegraphics[width=0.8\linewidth]{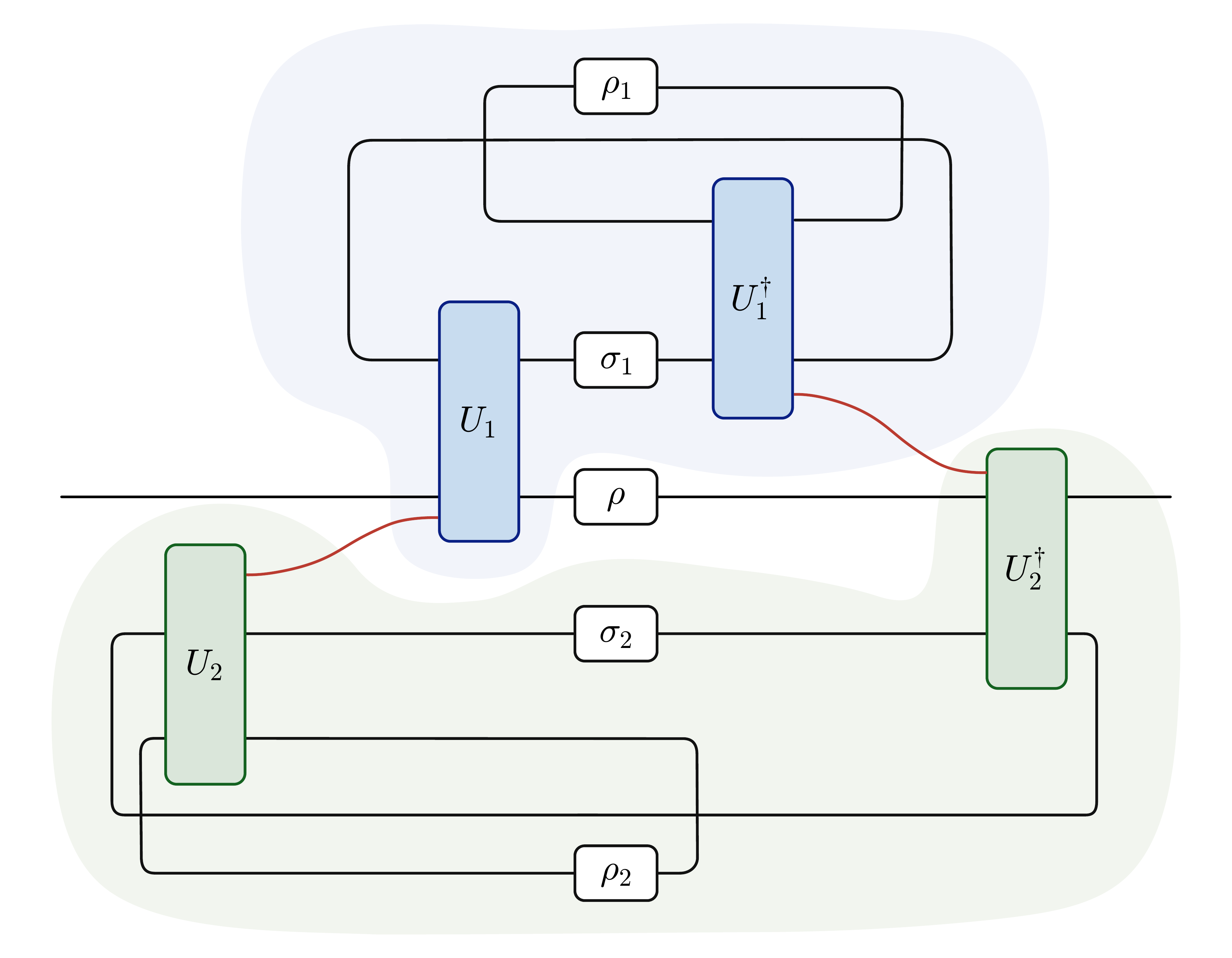}
    \caption{Finite-memory contraction pattern after tracing out the discarded registers and rearranging the TN. Relative to the memoryless case, the two contracted blocks are additionally connected by the forward and backward memory bonds.
    }
    \label{fig:C6_tensor_with_memory}
\end{figure}

\begin{figure}[htbp]
    \centering
    \includegraphics[width=0.44\linewidth]{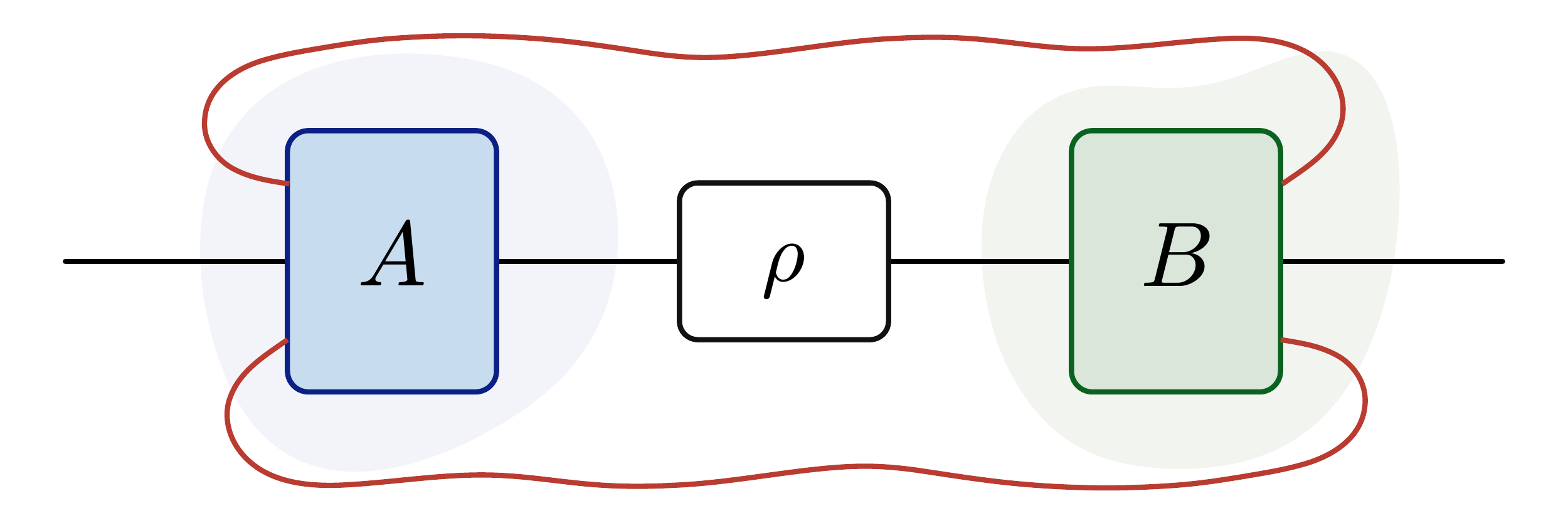}
    \caption{Simplified finite-memory network. The matrix-rank cut crosses the forward and backward memory bonds in addition to the pure system input $\rho$. Each memory bond has dimension $d_{\mathrm m}$.
    }
    \label{fig:C7_tensor_simple}
\end{figure}

Cut the two memory bonds and label their indices by
\begin{equation}
\alpha,\beta
=
1,\ldots,d_{\mathrm m},
\end{equation}
where $\alpha$ labels the forward memory bond and $\beta$ the backward memory bond. For a fixed pair $(\alpha,\beta)$, contracting all remaining tensors on the two sides of the cut gives two blocks, which we denote by $A_{\alpha\beta}$ and $B_{\alpha\beta}$. Summing over the memory indices therefore gives
\begin{equation}
M
=
\frac{1}{2}
\sum_{\alpha=1}^{d_{\mathrm m}}
\sum_{\beta=1}^{d_{\mathrm m}}
A_{\alpha\beta}\,
\rho\,
B_{\alpha\beta}.
\label{eq:appendix_memory_cut_decomposition}
\end{equation}
This is the finite-memory analogue of
Eq.~\eqref{eq:appendix_markovian_factorization}. In the memoryless case there is a single term, whereas here the pair $(\alpha,\beta)$ can take at most $d_{\mathrm m}^{2}$ values. For the pure system input used in this system-only argument
\begin{equation}
\rho
=
\ket{\psi}\!\bra{\psi},
\end{equation}
define
\begin{equation}
\ket{a_{\alpha\beta}}
:=
A_{\alpha\beta}\ket{\psi},
\qquad
\bra{b_{\alpha\beta}}
:=
\bra{\psi}B_{\alpha\beta}.
\label{eq:appendix_ab_vectors}
\end{equation}
Eq.~\eqref{eq:appendix_memory_cut_decomposition} then becomes
\begin{equation}
M
=
\frac{1}{2}
\sum_{\alpha=1}^{d_{\mathrm m}}
\sum_{\beta=1}^{d_{\mathrm m}}
\ket{a_{\alpha\beta}}
\!\bra{b_{\alpha\beta}}.
\label{eq:appendix_memory_outer_products}
\end{equation}
For each fixed pair $(\alpha,\beta)$, the corresponding term is a single outer product and therefore has rank at most one. Since there are at most $d_{\mathrm m}^{2}$ such pairs, subadditivity of matrix rank gives
\begin{align}
\operatorname{rank}(M)
&\leq
\sum_{\alpha=1}^{d_{\mathrm m}}
\sum_{\beta=1}^{d_{\mathrm m}}
\operatorname{rank}
\left(
\ket{a_{\alpha\beta}}
\!\bra{b_{\alpha\beta}}
\right)
\nonumber\\
&\leq
d_{\mathrm m}^{2}.
\label{eq:appendix_rank_le_dm2}
\end{align}
Hence
\begin{equation}
d_{\mathrm m}
\geq
\left\lceil
\sqrt{\operatorname{rank}(M)}
\right\rceil .
\label{eq:appendix_memory_lower_bound}
\end{equation}
This is the system-only form of the memory-dimension bound. The reference-assisted case required for Theorem~\ref{thm:memory_bound} is treated in the next subsection. For completeness, the same argument extends to a mixed system input of rank $r$. Writing $\rho$ as a sum of $r$ rank-one terms gives at most $r d_{\mathrm m}^{2}$ outer products, and therefore
\begin{equation}
\operatorname{rank}(M)
\leq
r\,d_{\mathrm m}^{2}.
\end{equation}
The pure-input result above corresponds to $r=1$.

\subsection{Reference-assisted memory bound and rank ceiling}
\label{subsec:appendix_reference_assisted}

We now restore the isolated reference $R$ and return to the interference matrix $M_{RS}$ used in the main text. Let
\begin{equation}
\ket{\Psi}_{RS}
\in
\mathcal H_R\otimes\mathcal H_S,
\qquad
\dim\mathcal H_S=d_S,
\label{eq:appendix_reference_spaces}
\end{equation}
be a pure reference-system input, with reduced states
\begin{equation}
\rho_S
:=
\operatorname{Tr}_{R}
\ket{\Psi}\!\bra{\Psi}_{RS},
\qquad
\rho_R
:=
\operatorname{Tr}_{S}
\ket{\Psi}\!\bra{\Psi}_{RS}.
\label{eq:appendix_reference_reduced_states}
\end{equation}
Because $\ket{\Psi}_{RS}$ is pure, the reduced states $\rho_S$ and $\rho_R$ have the same rank. Their common rank
\begin{equation}
s
:=
\operatorname{rank}(\rho_S)
=
\operatorname{rank}(\rho_R)
\label{eq:appendix_reference_schmidt_rank}
\end{equation}
is the Schmidt rank of $\ket{\Psi}_{RS}$. Only $S$ is routed through the pseudo-control circuit. The reference $R$ remains isolated throughout and is retained together with the final system output. The corresponding interference matrix is
\begin{equation}
M_{RS}
:=
\omega_{01}^{RS}.
\end{equation}
The reference changes the retained output space but does not introduce an additional bond across the memory cut. As a result, the memory-dimension bound derived above is unchanged, while the larger retained space can raise the maximum rank accessible to the protocol. We establish these two statements separately below.

\begin{figure}[htbp]
    \centering
    \includegraphics[width=0.82\linewidth]
    {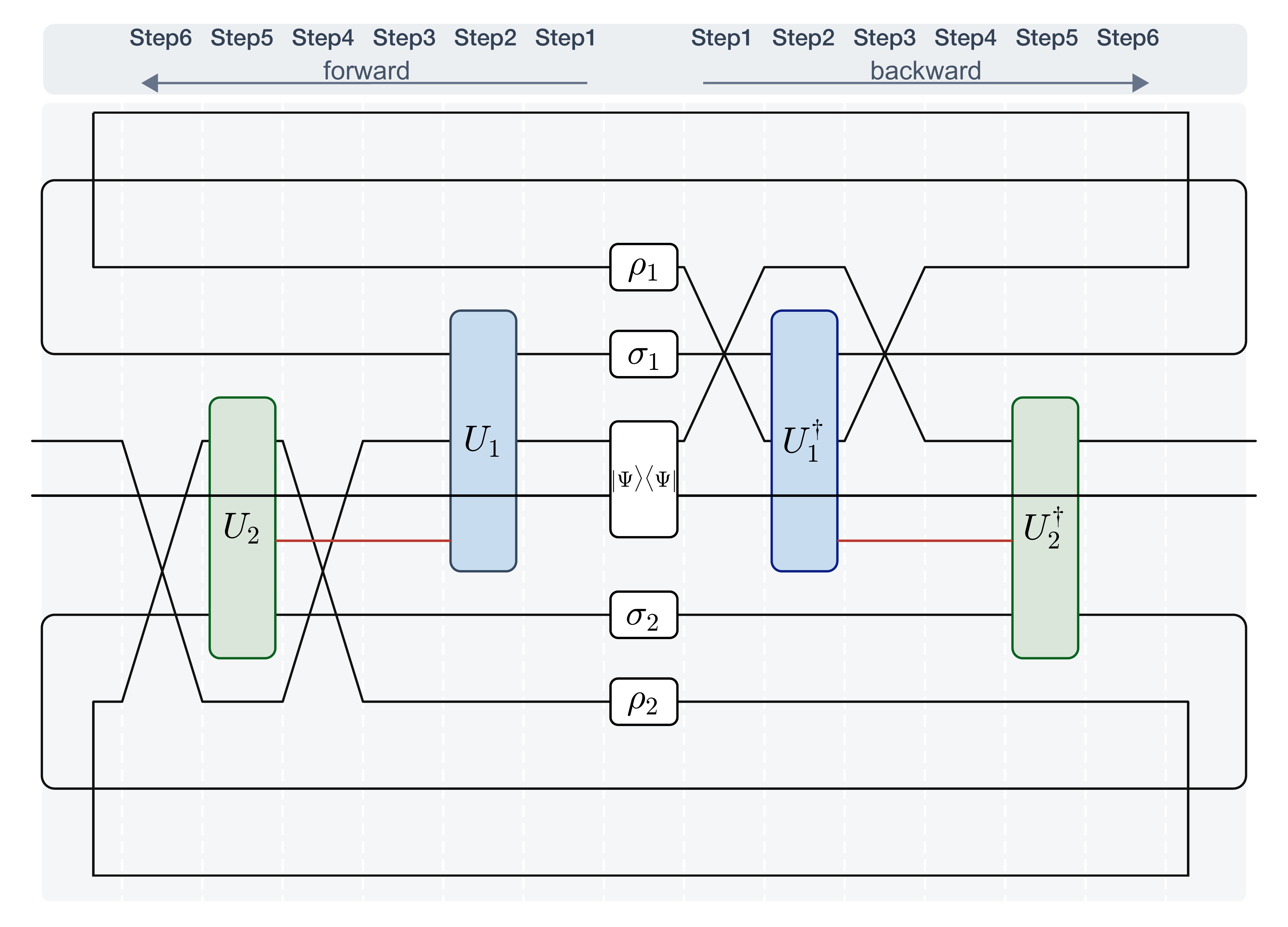}
    \caption{Step-by-step TN representation of the reference-assisted protocol. Only $S$ is routed through the pseudo-control circuit, while $R$ remains isolated and is retained for the final joint tomography.
    }
    \label{fig:appendix_reference_steps}
\end{figure}

\paragraph{Memory-dimension bound.}

The six visible circuit steps are unchanged from
Sec.~\ref{subsec:appendix_six_steps}. Relative to the system-only decomposition in Eq.~\eqref{eq:appendix_memory_cut_decomposition}, the pure system input and the two contracted blocks are replaced by
\begin{equation}
\ket{\psi}\!\bra{\psi}_{S}
\longrightarrow
\ket{\Psi}\!\bra{\Psi}_{RS},
\qquad
A_{\alpha\beta}
\longrightarrow
I_R\otimes A_{\alpha\beta},
\qquad
B_{\alpha\beta}
\longrightarrow
I_R\otimes B_{\alpha\beta}.
\label{eq:appendix_reference_replacements}
\end{equation}
Cutting the forward and backward memory bonds therefore gives
\begin{equation}
M_{RS}
=
\frac{1}{2}
\sum_{\alpha=1}^{d_{\mathrm m}}
\sum_{\beta=1}^{d_{\mathrm m}}
\left(
I_R\otimes A_{\alpha\beta}
\right)
\ket{\Psi}\!\bra{\Psi}_{RS}
\left(
I_R\otimes B_{\alpha\beta}
\right).
\label{eq:appendix_reference_memory_decomposition}
\end{equation}
Although the reduced system state $\rho_S$ can be mixed, the operator appearing in Eq.~\eqref{eq:appendix_reference_memory_decomposition} is the globally pure projector
$\ket{\Psi}\!\bra{\Psi}_{RS}$. For each fixed pair $(\alpha,\beta)$, define
\begin{equation}
\ket{a_{\alpha\beta}^{RS}}
:=
\left(
I_R\otimes A_{\alpha\beta}
\right)
\ket{\Psi}_{RS},
\qquad
\bra{b_{\alpha\beta}^{RS}}
:=
{}_{RS}\!\bra{\Psi}
\left(
I_R\otimes B_{\alpha\beta}
\right).
\label{eq:appendix_reference_ab_vectors}
\end{equation}
Then
\begin{equation}
M_{RS}
=
\frac{1}{2}
\sum_{\alpha=1}^{d_{\mathrm m}}
\sum_{\beta=1}^{d_{\mathrm m}}
\ket{a_{\alpha\beta}^{RS}}
\!\bra{b_{\alpha\beta}^{RS}}.
\label{eq:appendix_reference_rank_one_factorization}
\end{equation}
Thus, for each fixed pair $(\alpha,\beta)$, the corresponding contribution is a single outer product on the joint output space $R\otimes S$ and has rank at most one. Since the two memory indices each take at most $d_{\mathrm m}$ values, there are at most $d_{\mathrm m}^{2}$ such terms. Therefore
\begin{equation}
\operatorname{rank}(M_{RS})
\leq
d_{\mathrm m}^{2}.
\label{eq:appendix_reference_memory_bound}
\end{equation}
Equivalently
\begin{equation}
d_{\mathrm m}
\geq
\left\lceil
\sqrt{\operatorname{rank}(M_{RS})}
\right\rceil .
\label{eq:appendix_reference_memory_lower_bound}
\end{equation}
This proves Theorem~\ref{thm:memory_bound}. The isolated reference does not change this counting. It is retained throughout the circuit and does not create an additional bond connecting the two interactions across the memory cut. The only bonds added across this cut by finite memory are therefore the same forward and backward memory bonds already counted above.

\begin{figure}[htbp]
    \centering
    \includegraphics[width=0.8\linewidth]
    {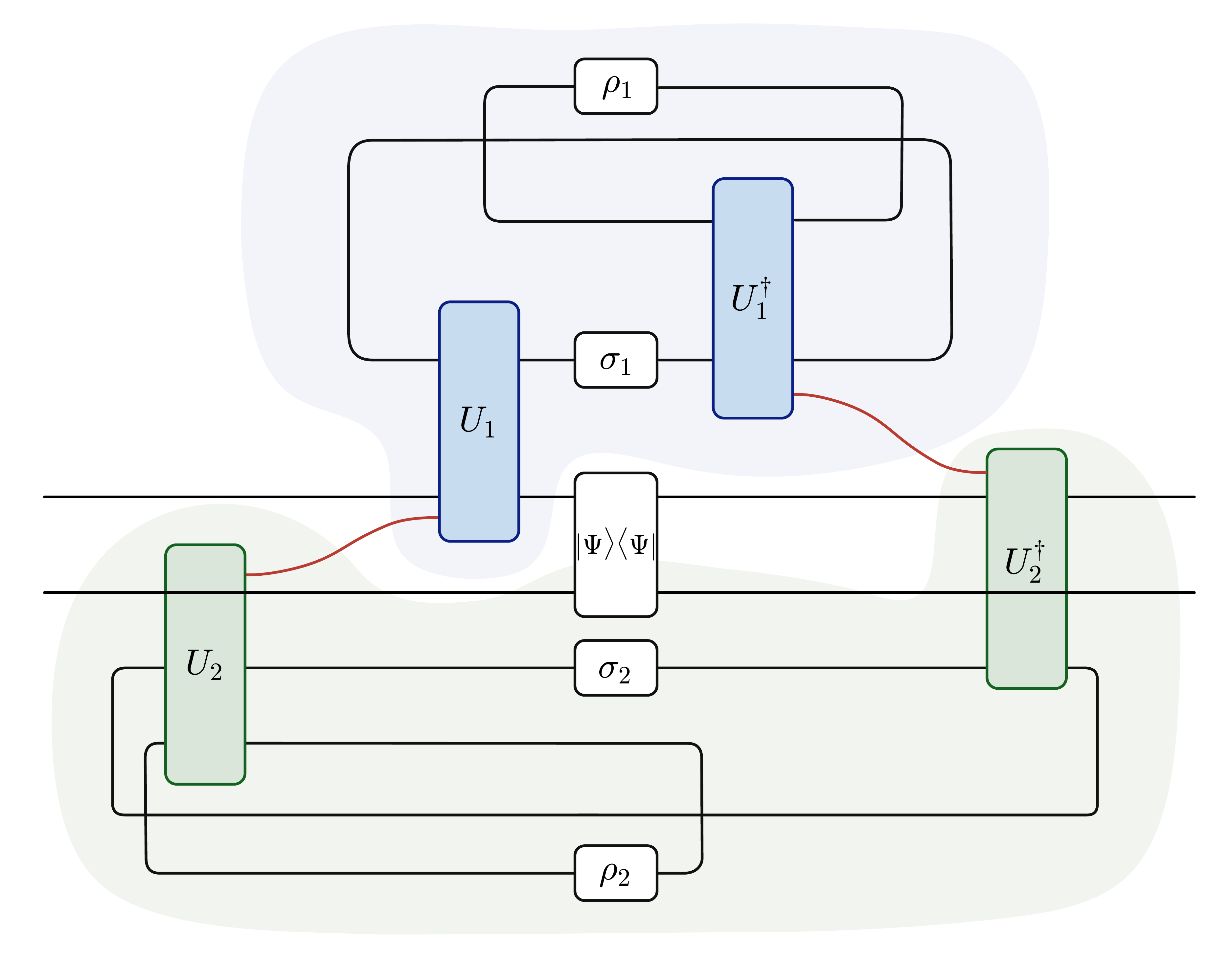}
    \caption{TN contraction pattern for the reference-assisted protocol. The global input operator $\ket{\Psi}\!\bra{\Psi}_{RS}$ is rank one, while the two red bonds are the ket and bra indices of the memory carrier connecting the two interactions. The isolated reference does not add a bond across the memory cut.
    }
    \label{fig:appendix_reference_tensor}
\end{figure}

\paragraph{Reference-assisted rank ceiling.}

The rank bound in Eq.~\eqref{eq:appendix_reference_memory_bound} is set by the physical memory connecting the two interactions. A separate limitation comes from the support of the retained $RS$ output. To see this, write the Schmidt decomposition of the input as
\begin{equation}
\ket{\Psi}_{RS}
=
\sum_{j=1}^{s}
\sqrt{\lambda_j}\,
\ket{j}_{R}
\ket{\phi_j}_{S},
\qquad
\lambda_j>0,
\qquad
\sum_{j=1}^{s}\lambda_j=1,
\label{eq:appendix_reference_schmidt_decomposition}
\end{equation}
where
$\{\ket{j}_{R}\}_{j=1}^{s}$
and
$\{\ket{\phi_j}_{S}\}_{j=1}^{s}$
are orthonormal sets. Let
\begin{equation}
P_R
:=
\sum_{j=1}^{s}
\ket{j}\!\bra{j}_{R}
\label{eq:appendix_reference_support_projector}
\end{equation}
be the projector onto the $s$-dimensional Schmidt support of $\rho_R$. Because no operation acts on $R$, the reference component of the output remains within this Schmidt support. Hence
\begin{equation}
M_{RS}
=
\left(
P_R\otimes I_S
\right)
M_{RS}
\left(
P_R\otimes I_S
\right).
\label{eq:appendix_reference_support_projection}
\end{equation}
The interference matrix is therefore supported on
\begin{equation}
\operatorname{supp}(\rho_R)
\otimes
\mathcal H_S,
\end{equation}
whose dimension is $s d_S$. It follows that
\begin{equation}
\operatorname{rank}(M_{RS})
\leq
s d_S.
\label{eq:appendix_reference_rank_ceiling}
\end{equation}

\begin{figure}[htbp]
    \centering
    \includegraphics[width=0.44\linewidth]
    {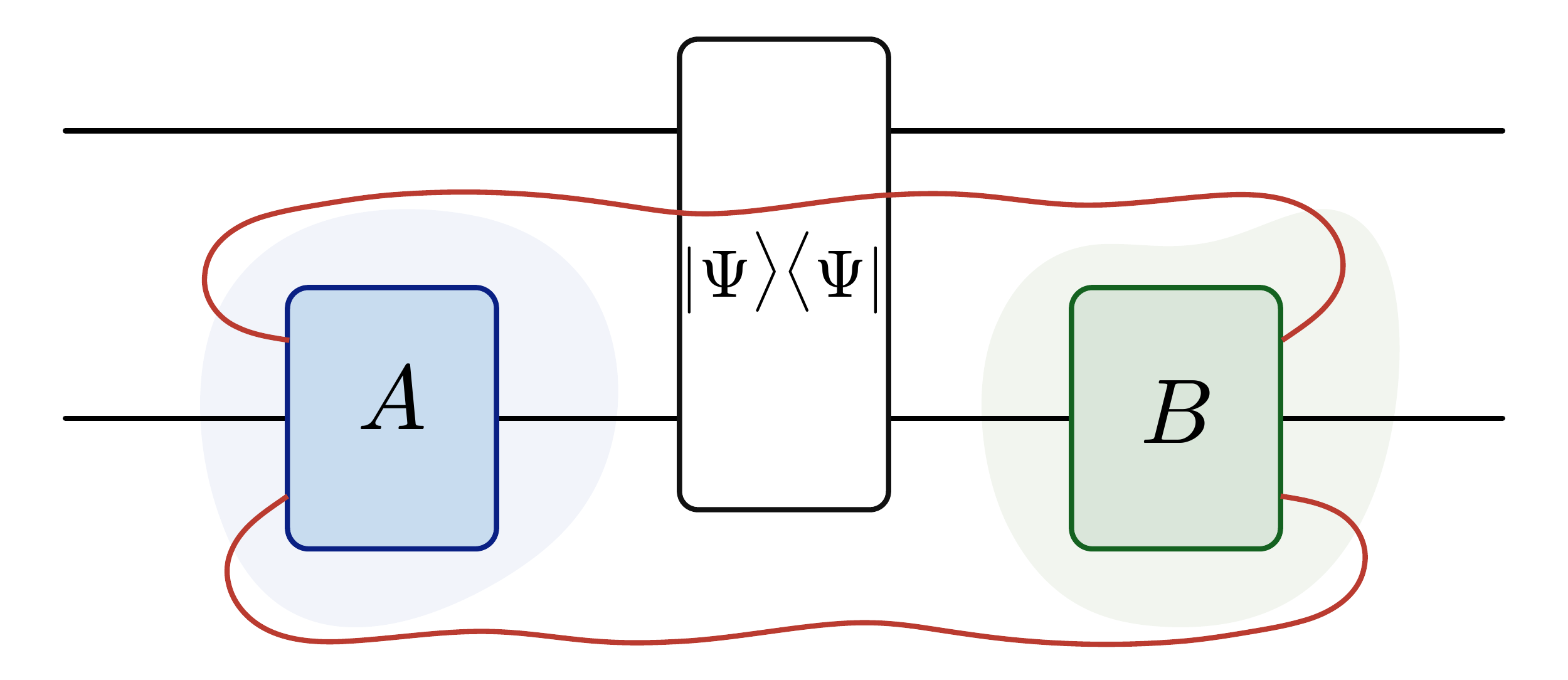}
    \caption{Simplified reference-assisted TN. Fixing the two memory indices leaves a single outer product on the joint output space $RS$. The reference output is restricted to its $s$-dimensional Schmidt support, so the retained joint support has dimension $s d_S$.
    }
    \label{fig:appendix_reference_simple}
\end{figure}
Without a reference, the interference matrix acts only on the $d_S$-dimensional system output and therefore has rank at most $d_S$. A reference of Schmidt rank $s$ enlarges the retained support to dimension $s d_S$, thereby raising the rank ceiling without changing the memory-dimension bound. For a full-Schmidt-rank input $s=d_S$, and
Eq.~\eqref{eq:appendix_reference_rank_ceiling} becomes
\begin{equation}
\operatorname{rank}(M_{RS})
\leq
d_S^{2}.
\label{eq:appendix_reference_full_schmidt_ceiling}
\end{equation}
A maximally entangled state is a natural choice in this case, as used in the numerical analysis, although equal Schmidt coefficients are not required for the support ceiling in Eq.~\eqref{eq:appendix_reference_full_schmidt_ceiling}. The two rank bounds above have different origins. The memory bound
$\operatorname{rank}(M_{RS})\leq d_{\mathrm m}^{2}$
is set by the dimension of the physical memory connecting the two interactions, whereas
$\operatorname{rank}(M_{RS})\leq s d_S$
is set by the support available in the retained reference-system output.

\section{Extensions and robustness of the memory bound}
\label{app:extensions-and-robustness}

The proof above uses a simplified setting to make the TN structure underlying the memory-dimension bound explicit. We now consider two extensions of this setting---open system-environment interactions and classical memory---and then examine the effect of residual system-memory dynamics during the interval. For clarity, the TN diagrams in this section suppress the isolated reference $R$ and show the corresponding system-only interference matrix $M$. Since $R$ does not participate in the two interactions, restoring it does not change any of the memory cuts considered below. The same cut arguments therefore extend directly to the reference-assisted interference matrix $M_{RS}$, as established in Sec.~\ref{subsec:appendix_reference_assisted}. We state the final memory bounds in terms of $M_{RS}$.

\subsection{System-environment interactions as quantum channels}
\label{app:unitary_to_channels}

The TN proof above represented the two interactions by unitaries $U_1$ and $U_2$. We now allow either or both interactions to be open-system channels. Physically, this includes local dissipation during a system-environment interaction, where the degrees of freedom involved in that interaction may also couple to additional local environmental modes.

\begin{figure}[htbp]
    \centering
    \includegraphics[width=0.6\linewidth]
    {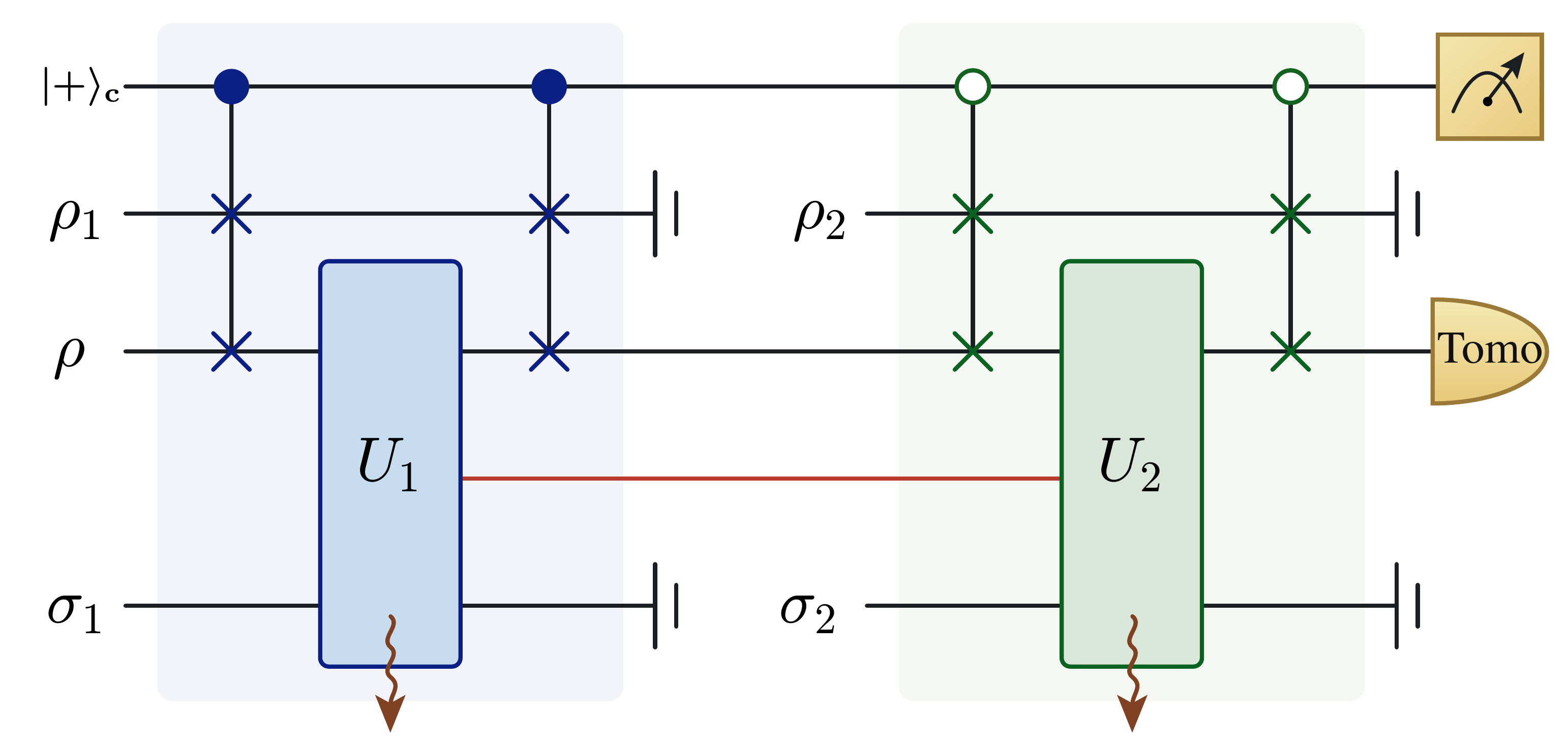}
    \caption{Circuit in which the two system-environment interactions include local dissipation and are therefore described by open-system channels. The isolated reference \(R\) is suppressed for clarity.
    }
    \label{fig:app_circuit_U1U2_dissipation}
\end{figure}

Each open interaction admits a Stinespring representation. For an input state $\varrho$
\begin{equation}
\mathcal E_i(\varrho)
=
\operatorname{Tr}_{e_i}
\left[
V_i
\left(
\varrho\otimes\eta_{e_i}
\right)
V_i^\dagger
\right],
\qquad
i=1,2.
\label{eq:open_stinespring}
\end{equation}
Here $e_i$ is a fresh local dilation degree of freedom associated with the $i$-th interaction, $\eta_{e_i}$ is its initial state, and $V_i$ is a unitary on the enlarged Hilbert space. The tensor-product form in Eq.~\eqref{eq:open_stinespring} assumes that $e_i$ is initially uncorrelated with the state entering that interaction. The corresponding TN representation is shown in Fig.~\ref{fig:app_tensor_dissipation}. Relative to the unitary case, $U_i$ is replaced by the enlarged unitary $V_i$, and the additional index associated with $e_i$ is introduced and traced out locally within the same interaction.

\begin{figure}[htbp]
    \centering
    \includegraphics[width=0.8\linewidth]
    {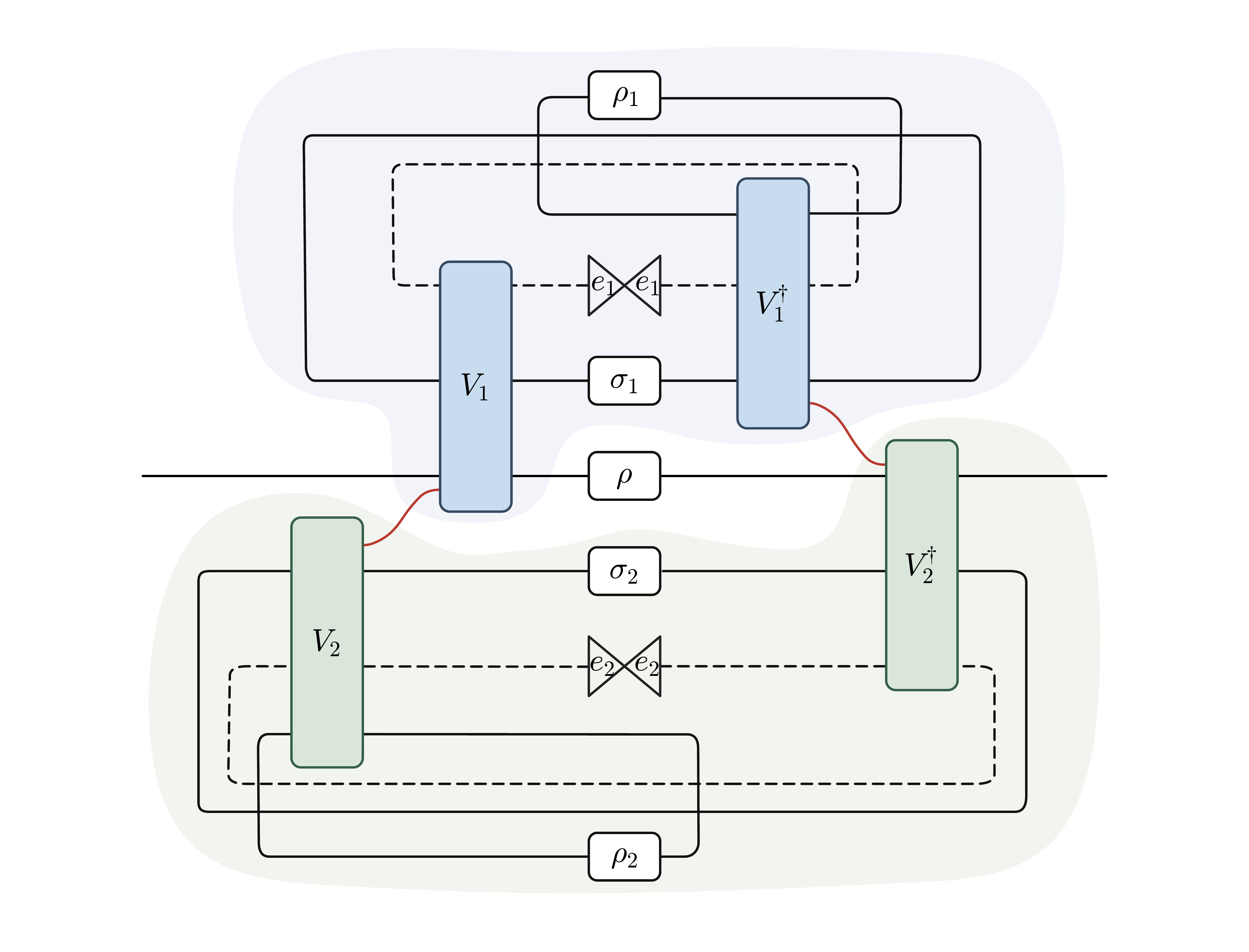}
    \caption{Stinespring representation of the two open interactions. The unitaries $U_1$ and $U_2$ are replaced by enlarged unitaries $V_1$ and $V_2$. The dilation degrees of freedom $e_1$ and $e_2$ are initialised and traced out locally within their respective interactions, so their indices do not cross the memory cut.
    }
    \label{fig:app_tensor_dissipation}
\end{figure}

Because each dilation degree of freedom is local to a single interaction, its TN indices close within the corresponding contracted block. The local dilation therefore does not introduce an additional bond across the memory cut. Replacing $U_i$ by $V_i$ consequently leaves the memory-bond counting unchanged. Restoring the isolated reference does not change this cut. Hence the reference-assisted interference matrix still obeys
\begin{equation}
\operatorname{rank}(M_{RS})
\leq
d_{\mathrm m}^{2},
\end{equation}
and therefore
\begin{equation}
d_{\mathrm m}
\geq
\left\lceil
\sqrt{\operatorname{rank}(M_{RS})}
\right\rceil .
\label{eq:open_reference_memory_bound}
\end{equation}
Thus Theorem~\ref{thm:memory_bound} remains valid when either or both system-environment interactions are open-system channels. The locality of the dilation degrees of freedom is essential. If a degree of freedom created during the first interaction is retained and can later influence the second interaction, it is no longer local to the first interaction. It must instead be included in the memory carrier and therefore contributes to the operational memory dimension $d_{\mathrm m}$.

\subsection{Classical memory}
\label{app:classical-memory}

The memory connecting the two interactions need not be quantum. A purely classical memory can be represented by an orthogonal label
$ 
k=1,\ldots,d_{\rm cl},
$
that is retained after the first interaction and remains available to the second~\cite{Giarmatzi2021Witnessing,Taranto2024ClassicalMemory}. The label may be generated or updated during the first interaction; what matters is that only classical information is carried between the two interactions. As a simple example, a unitary $U_k$ may be selected randomly and the same classical value $k$ used in both interactions. 
Each interaction then appears individually as a random-unitary mixture, while the shared value of $k$ correlates the two interactions in time. For a classical memory, the forward and backward memory indices carry the same classical value $k$. Thus, the pair of independent memory indices in the quantum case is replaced by a single classical branch index.

Let $p_k$ be the probability of the classical memory value $k$. For the pure reference--system input $\ket{\Psi}_{RS}$, conditioning the two contracted blocks on $k$ gives
\begin{equation}
M_{RS}
=
\sum_{k=1}^{d_{\rm cl}}
p_k
\left(
I_R\otimes A_k
\right)
\ket{\Psi}\!\bra{\Psi}_{RS}
\left(
I_R\otimes B_k
\right).
\label{eq:classical_memory_decomposition}
\end{equation}
Define
\begin{equation}
\ket{u_k}
:=
\sqrt{p_k}
\left(
I_R\otimes A_k
\right)
\ket{\Psi}_{RS},
\qquad
\bra{v_k}
:=
\sqrt{p_k}\,
{}_{RS}\!\bra{\Psi}
\left(
I_R\otimes B_k
\right).
\end{equation}
Then
\begin{equation}
M_{RS}
=
\sum_{k=1}^{d_{\rm cl}}
\ket{u_k}\!\bra{v_k}.
\label{eq:classical_memory_outer_products}
\end{equation}
For each fixed classical value $k$, the corresponding term is a single outer product and therefore has rank at most one. Since there are at most $d_{\rm cl}$ classical values
\begin{equation}
\operatorname{rank}(M_{RS})
\leq
d_{\rm cl}.
\label{eq:classical_memory_bound}
\end{equation}
Thus, if
$\operatorname{rank}(M_{RS})=r$,
any purely classical realisation of this form requires at least $r$ distinguishable memory states
\begin{equation}
d_{\rm cl}\geq r.
\end{equation}
A rank larger than one can therefore arise entirely from classical temporal memory; quantum coherence in the memory is not required. Conversely, rank one does not exclude the presence of classical memory.

\subsection{Residual system-memory dynamics during the interval}
\label{app:intermediate_dynamics}

The protocol assumes that the system is decoupled from the memory during the interval between the two interactions. In practice, this switch-off may be imperfect and leave a residual system-memory coupling. Here we examine how the resulting residual dynamics affects the memory-dimension bound. Evolution acting only on the memory during the interval does not cause a problem: it can be absorbed into the propagation of the memory carrier between the two interactions. The relevant modification is an additional interaction involving the system itself. As throughout this section, the TN diagrams suppress the isolated reference $R$ for clarity. Restoring $R$ does not change the memory cuts, so the final bounds are stated for the reference-assisted interference matrix $M_{RS}$.

\paragraph{Residual unitary interaction.}

We first suppose that the residual dynamics is unitary
\begin{equation}
\mathcal U_3(\varrho)
=
U_3\varrho U_3^\dagger .
\label{eq:residual_unitary}
\end{equation}

\begin{figure}[htbp]
    \centering
    \includegraphics[width=0.6\linewidth]
    {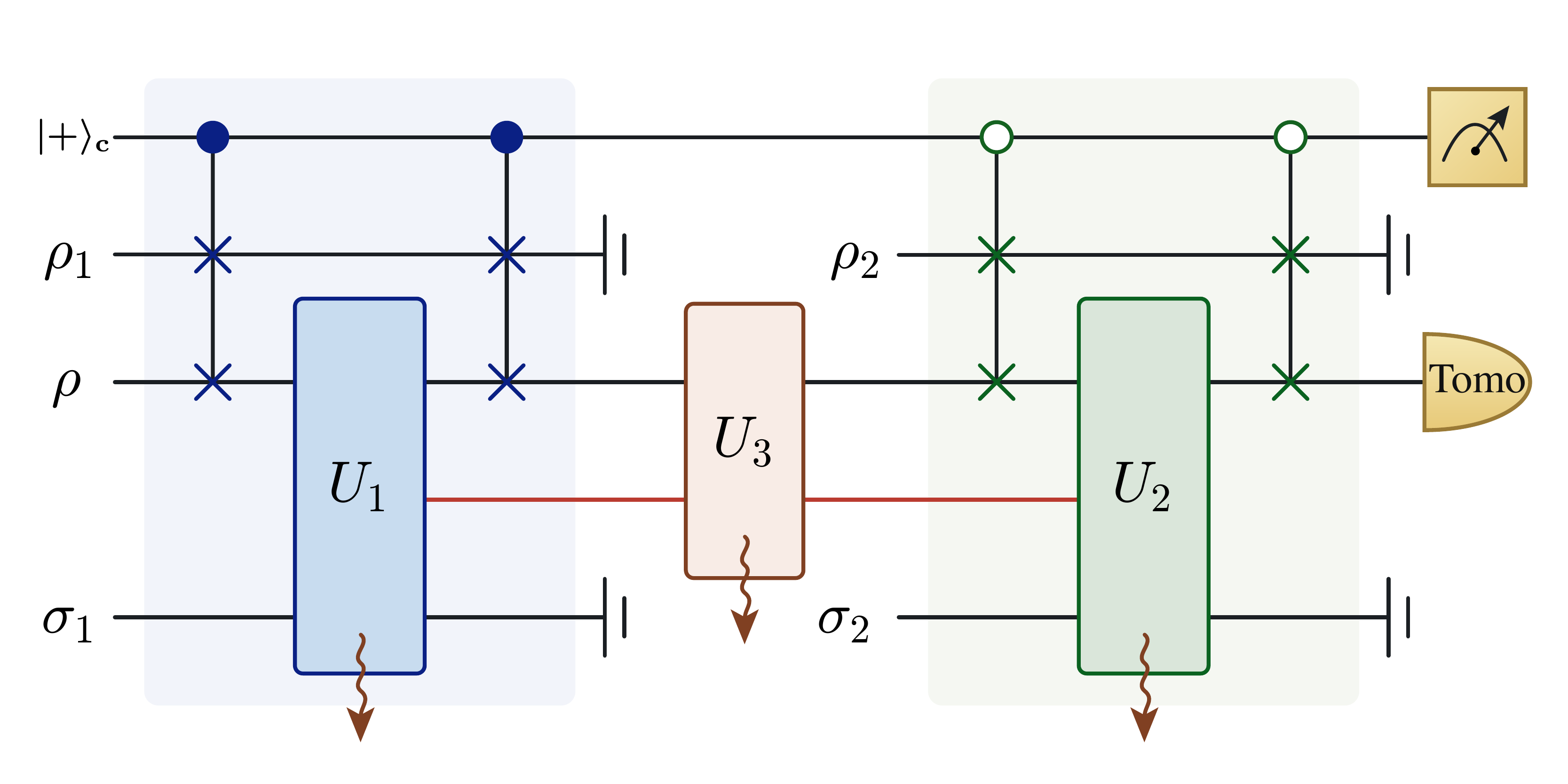}
    \caption{Circuit with a residual unitary system-memory interaction $U_3$ during the interval. The isolated reference $R$ is suppressed for clarity.
    }
    \label{fig:app_circuit_U3_interaction}
\end{figure}
The corresponding circuit is shown in
Fig.~\ref{fig:app_circuit_U3_interaction}. A residual unitary does not change the contraction structure responsible for the memory bound. In the system-only TN shown in Fig.~\ref{fig:direct_truncation_additional}(a), $U_3$ can be absorbed into one contracted block and $U_3^\dagger$ into the other. No additional bond is introduced across the memory cut. Restoring the isolated reference therefore gives
\begin{equation}
\operatorname{rank}
\left(
M_{RS}^{\mathcal U}
\right)
\leq
d_{\mathrm m}^{2},
\label{eq:residual_unitary_rank_bound}
\end{equation}
where $M_{RS}^{\mathcal U}$ denotes the interference matrix obtained in the presence of the residual unitary channel $\mathcal U_3$. Hence
\begin{equation}
d_{\mathrm m}
\geq
\left\lceil
\sqrt{
\operatorname{rank}
\left(
M_{RS}^{\mathcal U}
\right)
}
\right\rceil.
\end{equation}

Importantly, this conclusion does not require the residual unitary interaction to be weak. Any residual unitary can be absorbed into the two contracted blocks without changing the number or dimension of the bonds crossing the memory cut. The exact memory-dimension bound in Theorem \ref{thm:memory_bound} therefore still remains valid.

\paragraph{Residual open dynamics.}

The situation is different if the residual dynamics is an open-system channel $\mathcal E_3$. A Stinespring representation of $\mathcal E_3$ can introduce an additional dilation index that connects the two sides of the contraction, as illustrated in Fig.~\ref{fig:direct_truncation_additional}(b). In this case, the original memory-cut argument no longer directly implies an exact algebraic-rank bound of the form
$\operatorname{rank}(M_{RS})\leq d_{\mathrm m}^{2}$.

\begin{figure}[htbp]
    \centering

    \begin{minipage}{0.495\linewidth}
        \centering
        \includegraphics[width=\linewidth]
        {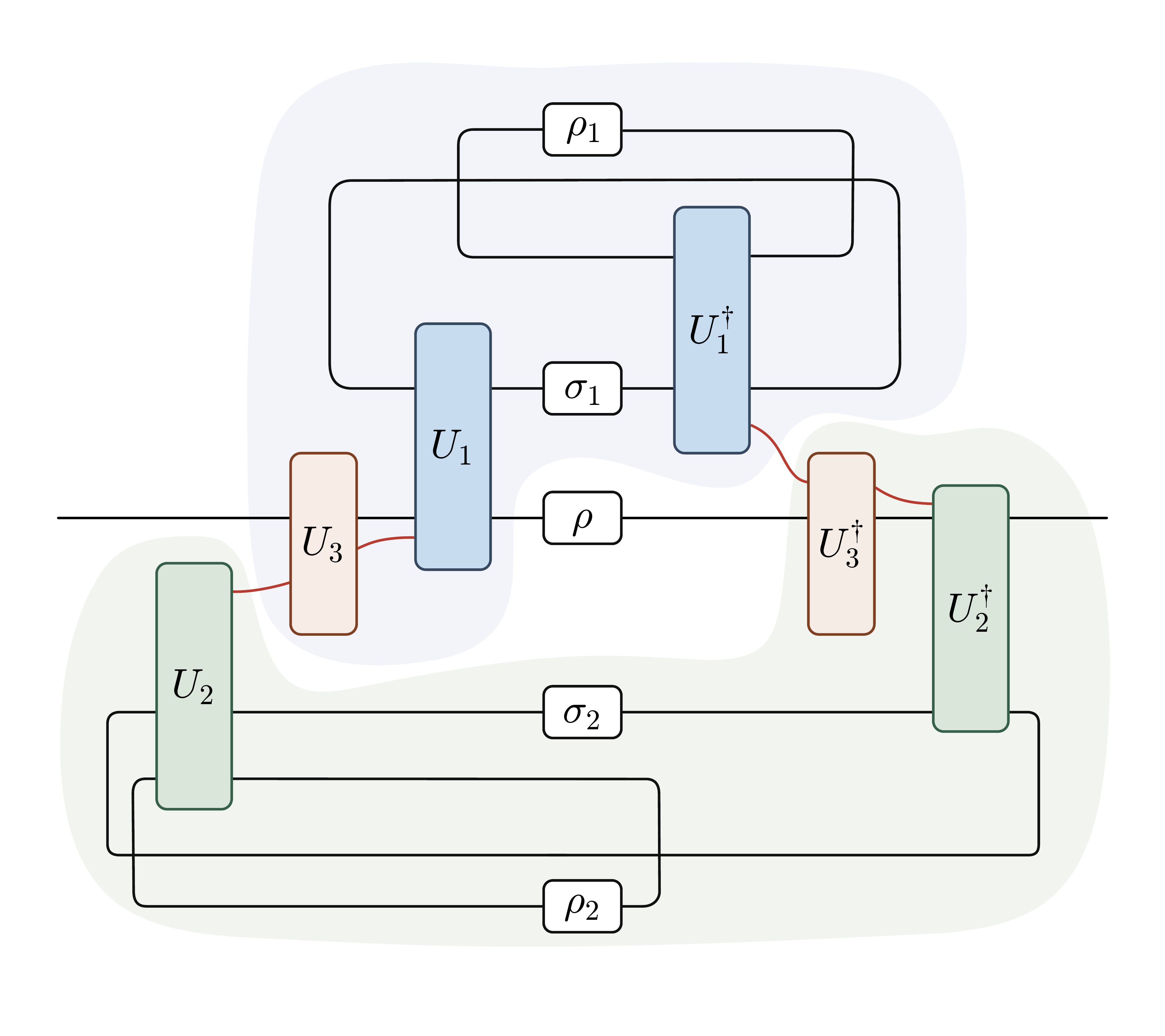}
        \textbf{(a)} Residual unitary coupling $U_3$.
    \end{minipage}
    \begin{minipage}{0.495\linewidth}
        \centering
        \includegraphics[width=\linewidth]
        {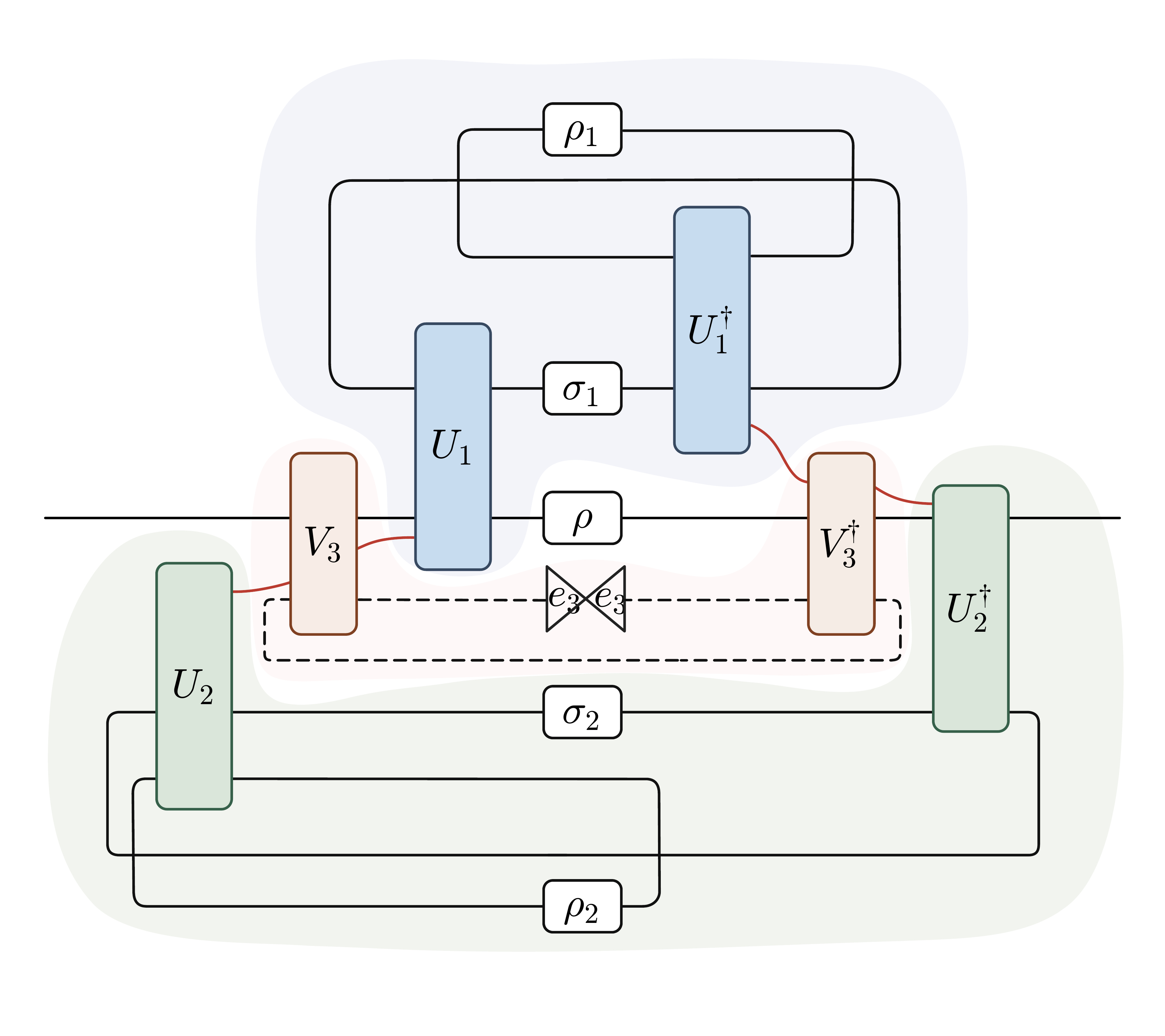}
        \textbf{(b)} Residual open channel $\mathcal E_3$.
    \end{minipage}

    \caption{TN representations of residual system-memory dynamics during the interval.
    (a) A residual unitary can be absorbed into the two contracted blocks without introducing an additional bond across the memory cut.
    (b) A residual open channel \(\mathcal E_3\), represented by a Stinespring dilation \(V_3\), can introduce an additional dilation index across the same cut, so the exact unitary cut argument no longer applies directly. 
    }
    \label{fig:direct_truncation_additional}
\end{figure}

A modified bound can nevertheless be obtained if the residual open dynamics is independently known to be close to a unitary channel. Suppose that, for some residual unitary channel $\mathcal U_3$
\begin{equation}
\frac{1}{2}
\left\|
\mathcal E_3-\mathcal U_3
\right\|_{\diamond}
\leq
\delta .
\label{eq:residual_diamond_bound}
\end{equation}
Here $\delta$ is an independently characterised upper bound on the deviation of the residual open dynamics from the reference unitary channel. The strongest bound of this form is obtained by choosing a unitary channel that minimises the diamond-norm distance, but the argument below applies to any independently established bound in Eq.~\eqref{eq:residual_diamond_bound}.

Let
$M_{RS}^{\mathcal E}$
and
$M_{RS}^{\mathcal U}$
denote the interference matrices obtained with residual dynamics
$\mathcal E_3$
and
$\mathcal U_3$,
respectively. Eq.~\eqref{eq:residual_diamond_bound} limits how far these two matrices can differ. In particular, singular values of
$M_{RS}^{\mathcal E}$
that lie above the scale $\delta$ must correspond to nonzero singular values of
$M_{RS}^{\mathcal U}$.
This gives the following modified memory bound.

\begin{theorem}[Memory bound under residual open dynamics]
\label{thm:robust_memory_bound}

Suppose that the residual open channel $\mathcal E_3$ satisfies
\begin{equation}
\frac{1}{2}
\left\|
\mathcal E_3-\mathcal U_3
\right\|_{\diamond}
\leq
\delta
\end{equation}
for a unitary channel $\mathcal U_3$. Then
\begin{equation}
\#
\left\{
k:
s_k
\left(
M_{RS}^{\mathcal E}
\right)
>
\delta
\right\}
\leq
\operatorname{rank}
\left(
M_{RS}^{\mathcal U}
\right)
\leq
d_{\mathrm m}^{2}.
\label{eq:robust_rank_chain}
\end{equation}
Consequently
\begin{equation}
d_{\mathrm m}
\geq
\left\lceil
\sqrt{
\#
\left\{
k:
s_k
\left(
M_{RS}^{\mathcal E}
\right)
>
\delta
\right\}
}
\right\rceil .
\label{eq:robust_memory_bound}
\end{equation}
\end{theorem}

\begin{proof}

Keep all other parts of the protocol fixed, and let
$\omega_{cRS}^{\mathcal E}$
and
$\omega_{cRS}^{\mathcal U}$
denote the final retained states obtained with
$\mathcal E_3$
and
$\mathcal U_3$,
respectively. The preparation, the rest of the circuit, and the final partial trace are the same in the two cases. Contractivity therefore gives
\begin{equation}
\frac{1}{2}
\left\|
\omega_{cRS}^{\mathcal E}
-
\omega_{cRS}^{\mathcal U}
\right\|_1
\leq
\frac{1}{2}
\left\|
\mathcal E_3-\mathcal U_3
\right\|_{\diamond}
\leq
\delta .
\label{eq:residual_output_trace_bound}
\end{equation}
Define
\begin{equation}
\Delta\omega
:=
\omega_{cRS}^{\mathcal E}
-
\omega_{cRS}^{\mathcal U},
\qquad
\Delta M
:=
M_{RS}^{\mathcal E}
-
M_{RS}^{\mathcal U}.
\end{equation}
In the computational basis of the control, $\Delta M$ is the off-diagonal block of the Hermitian operator $\Delta\omega$. Its off-diagonal part can be written as
\begin{equation}
\Delta\omega_{\mathrm{off}}
=
\begin{pmatrix}
0 & \Delta M \\
\Delta M^\dagger & 0
\end{pmatrix}.
\label{eq:residual_off_diagonal_block}
\end{equation}
Equivalently, if
\begin{equation}
Z_c
=
\ket{0}\!\bra{0}_{c}
-
\ket{1}\!\bra{1}_{c},
\end{equation}
then
\begin{equation}
\Delta\omega_{\mathrm{off}}
=
\frac{1}{2}
\left[
\Delta\omega
-
\left(
Z_c\otimes I_{RS}
\right)
\Delta\omega
\left(
Z_c\otimes I_{RS}
\right)
\right].
\end{equation}
By the triangle inequality and unitary invariance of the trace norm
\begin{equation}
\left\|
\Delta\omega_{\mathrm{off}}
\right\|_1
\leq
\left\|
\Delta\omega
\right\|_1.
\end{equation}
The block matrix in
Eq.~\eqref{eq:residual_off_diagonal_block}
has singular values equal to those of $\Delta M$, each appearing twice. Hence
\begin{equation}
\left\|
\Delta\omega_{\mathrm{off}}
\right\|_1
=
2
\left\|
\Delta M
\right\|_1.
\end{equation}
Combining these relations with
Eq.~\eqref{eq:residual_output_trace_bound} gives
\begin{equation}
\left\|
M_{RS}^{\mathcal E}
-
M_{RS}^{\mathcal U}
\right\|_{\mathrm{op}}
\leq
\left\|
M_{RS}^{\mathcal E}
-
M_{RS}^{\mathcal U}
\right\|_1
\leq
\delta .
\label{eq:residual_matrix_bound}
\end{equation}

The standard perturbation bound for singular values therefore gives, for every $k$,
\begin{equation}
\left|
s_k
\left(
M_{RS}^{\mathcal E}
\right)
-
s_k
\left(
M_{RS}^{\mathcal U}
\right)
\right|
\leq
\delta .
\label{eq:residual_singular_value_bound}
\end{equation}
Thus, whenever
\begin{equation}
s_k
\left(
M_{RS}^{\mathcal E}
\right)
>
\delta ,
\end{equation}
we have
\begin{equation}
s_k
\left(
M_{RS}^{\mathcal U}
\right)
\geq
s_k
\left(
M_{RS}^{\mathcal E}
\right)
-
\delta
>
0.
\end{equation}
Every singular value of
$M_{RS}^{\mathcal E}$
that lies above the independently characterised scale $\delta$ therefore corresponds to a nonzero singular value of
$M_{RS}^{\mathcal U}$. Hence
\begin{equation}
\#
\left\{
k:
s_k
\left(
M_{RS}^{\mathcal E}
\right)
>
\delta
\right\}
\leq
\operatorname{rank}
\left(
M_{RS}^{\mathcal U}
\right).
\end{equation}
Since $\mathcal U_3$ is unitary,
Eq.~\eqref{eq:residual_unitary_rank_bound} gives
\begin{equation}
\operatorname{rank}
\left(
M_{RS}^{\mathcal U}
\right)
\leq
d_{\mathrm m}^{2}.
\end{equation}
Combining the two inequalities proves
Eq.~\eqref{eq:robust_memory_bound}.
\end{proof}

The two cases above should be distinguished. A residual unitary interaction leaves the exact algebraic-rank bound unchanged, regardless of its strength. A residual open channel can instead invalidate the direct exact-rank argument, but an independently characterised deviation $\delta$ from a unitary channel yields the modified certificate in Eq.~\eqref{eq:robust_memory_bound}. If no independent bound on the residual open dynamics is available, this argument does not provide a rigorous value of $\delta$, and an arbitrary numerical cutoff cannot replace such a bound. Conversely, if an additional degree of freedom generated during the interval is retained and can later influence the second interaction, it should not be treated as a perturbative error. It forms part of the memory carrier and must therefore be included in the operational memory dimension $d_{\mathrm m}$.

\section{Memory-dimension bound in the quantum-comb representation}
\label{sec:comb-memory-rank}

The preceding section proves the memory-dimension bound using a tensor-network representation of the pseudo-control circuit.
Here we give the same argument in the quantum-comb representation~\cite{Bisio2012MemoryCost,Chiribella2008Supermaps}, which may be more natural for readers familiar with multitime processes.
Fig.~\ref{fig:comb_tensor_to_comb} shows how the circuit tensor network is rearranged and coarse-grained into the corresponding comb representation.
In this representation, the memory connecting the two interactions appears as a pair of ket and bra indices, whose dimensions directly give the same rank bound.

Fig.~\ref{fig:comb_tensor_to_comb}(a) shows the circuit TN used in the preceding proof.
Panel~(b) rearranges the same contraction so that the two interaction regions and the memory bonds connecting them are more explicit; no tensor contraction is changed in this step.
Contracting the tensors within the two regions gives the coarse-grained network in panel~(c).
The two auxiliary system inputs $\rho_1$ and $\rho_2$, together with the routed system input $\rho$, are still attached, so this network represents the system-only interference matrix $M$ for those particular inputs.
Allowing the two interactions to be open changes only the tensors inside the two coarse-grained regions: as shown in
Sec.~\ref{app:unitary_to_channels}, local dilation degrees of freedom do not cross the memory cut.
Opening the legs occupied by the three fixed inputs therefore gives the corresponding two-interaction comb shown in Fig.~\ref{fig:comb_memory_cut}(a).

\begin{figure}[htbp]
    \centering
    \includegraphics[width=0.98\linewidth]
    {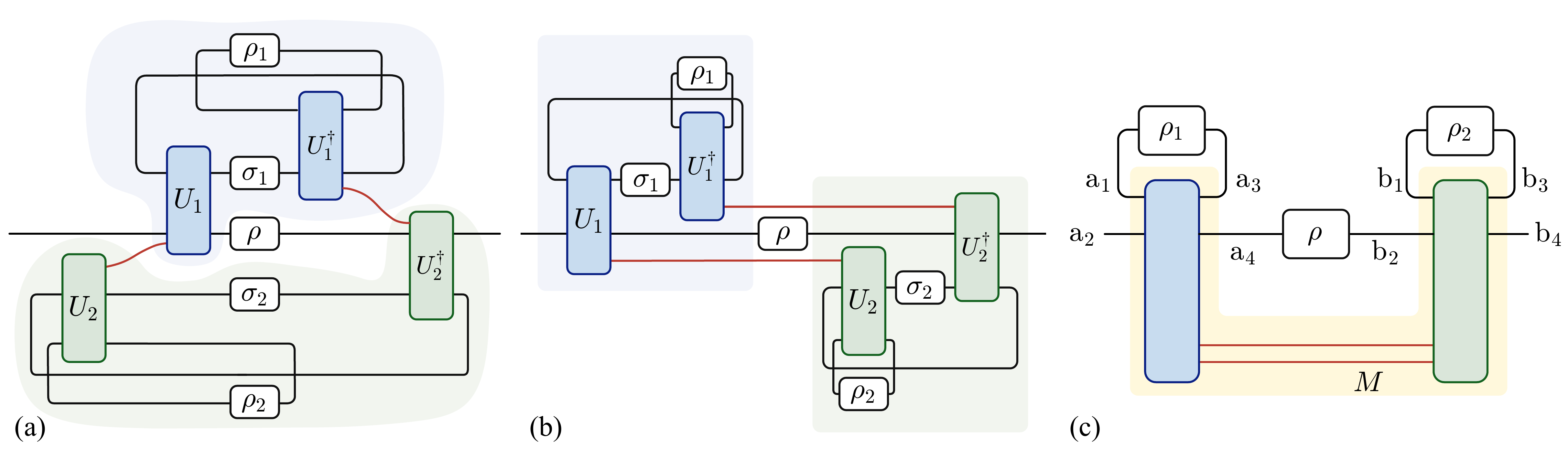}
    \caption{
    From the circuit TN to the fixed-input network.
    (a) Circuit TN of the pseudo-control construction, with the isolated reference $R$ suppressed for clarity.
    (b) Equivalent rearrangement of the same tensor contraction, making the two interaction regions and the memory bonds between them explicit.
    (c) Coarse-grained network obtained by contracting the tensors within the two regions.
    The auxiliary system inputs $\rho_1$ and $\rho_2$ and the routed system input $\rho$ remain attached.
    Opening the corresponding input legs gives the comb in Fig.~\ref{fig:comb_memory_cut}(a).
    }
    \label{fig:comb_tensor_to_comb}
\end{figure}

Fig.~\ref{fig:comb_memory_cut}(a) shows the Choi operator \(\Upsilon_{2:0}\) of the two-interaction comb, obtained by leaving these input legs open. It therefore represents the same two-interaction process before these particular system inputs are inserted.
To expose the temporal memory cut between the two interactions, we group the open indices on the two sides of the network into a composite row index $\mathbf a$ and a composite column index $\mathbf b$.
This gives the matricisation $\widetilde{\Upsilon}$ shown in panel~(b).
The bending of the open legs only regroups tensor indices to form a matrix; it does not represent an additional physical operation.
In this form, the dashed cut crosses the two red bonds associated with the memory carried between the first and second interactions.

\begin{figure}[htbp]
    \centering
    \includegraphics[width=0.72\linewidth]
    {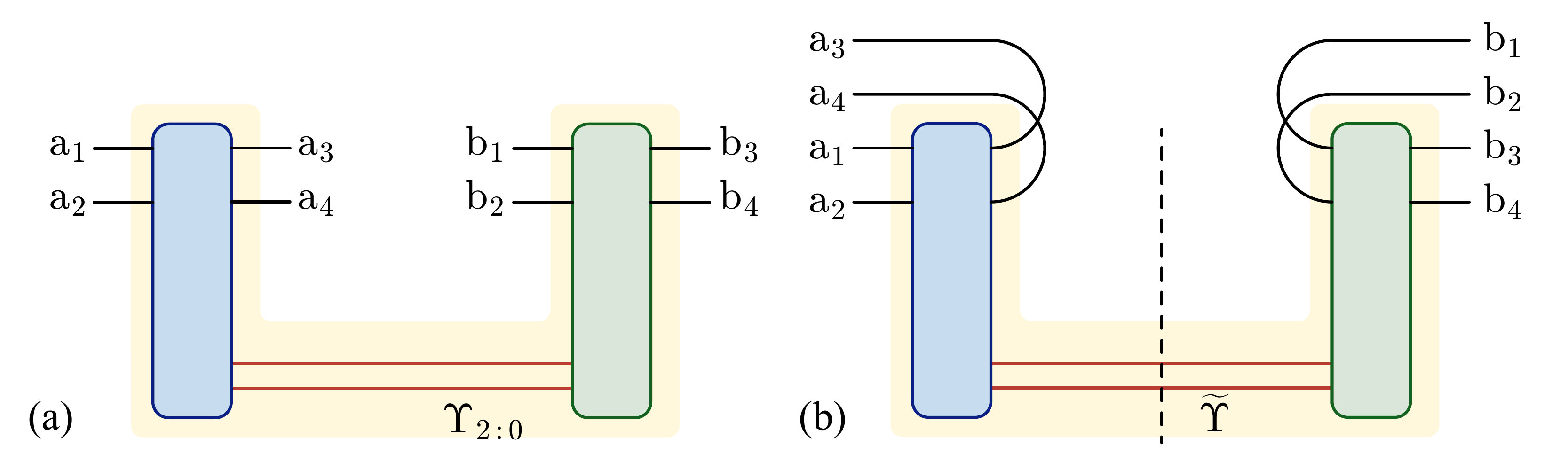}
    \caption{Quantum-comb representation of the memory-dimension bound.
    (a) Two-interaction comb tensor $\Upsilon_{2:0}$ obtained by opening the input legs of the fixed-input network in Fig.~\ref{fig:comb_tensor_to_comb}(c).
    (b) Left-right matricisation $\widetilde{\Upsilon}$ of the same comb.
    The bending of the open legs only regroups tensor indices, while the dashed line separates the two interaction regions and crosses the two memory bonds.
    The two red bonds are the ket and bra indices of the same physical memory, each of dimension $d_{\mathrm m}$.
    The isolated reference $R$ is suppressed for clarity.
    }
    \label{fig:comb_memory_cut}
\end{figure}

Let the physical memory have Hilbert-space dimension $d_{\mathrm m}$.
In the operator representation, this memory contributes one ket index and one bra index, each of dimension $d_{\mathrm m}$.
We label these indices by
\begin{equation}
\mu,\nu
=
1,\ldots,d_{\mathrm m}.
\label{eq:comb_memory_indices}
\end{equation}
Cutting the two memory bonds in Fig.~\ref{fig:comb_memory_cut}(b) gives
\begin{equation}
\widetilde{\Upsilon}_{\mathbf a,\mathbf b}
=
\sum_{\mu,\nu=1}^{d_{\mathrm m}}
A_{\mathbf a;\mu\nu}\,
B_{\mu\nu;\mathbf b}.
\label{eq:comb-factorization}
\end{equation}
Eq.~\eqref{eq:comb-factorization} factorises $\widetilde{\Upsilon}$ through the pair of memory indices $(\mu,\nu)$, which span a space of dimension at most $d_{\mathrm m}^{2}$.
Hence
\begin{equation}
\operatorname{rank}
\left(
\widetilde{\Upsilon}
\right)
\leq
d_{\mathrm m}^{2}.
\label{eq:upsilon-rank-bound}
\end{equation}
The square arises because the operator representation contains the ket and bra indices of the same $d_{\mathrm m}$-dimensional physical memory.
Equivalently, \(\operatorname{rank}(\widetilde{\Upsilon})\) is the operator-space bond rank of the comb Choi operator across this temporal cut, and is therefore bounded by \(d_{\mathrm m}^{2}\). 
The same doubled memory space appears as the internal operator-space bond in matrix-product-operator representations of process tensors~\cite{Pollock2018ProcessTensor}. 

We finally reinsert the two auxiliary system inputs and the pure reference-system input $\ket{\Psi}_{RS}$ used in the main text.
Because $\ket{\Psi}\!\bra{\Psi}_{RS}$ is rank one, its ket and bra amplitudes can be absorbed into the two sides of the matricisation.
The isolated reference is retained and does not introduce an additional bond between the two interactions.
Together with the auxiliary inputs, these contractions can therefore be absorbed into fixed left and right maps, giving
\begin{equation}
M_{RS}
=
L_{\Psi}\,
\widetilde{\Upsilon}\,
R_{\Psi}.
\label{eq:comb_MRS_factorisation}
\end{equation}
Multiplication by fixed matrices cannot increase rank, and hence
\begin{equation}
\operatorname{rank}(M_{RS})
\leq
\operatorname{rank}
\left(
\widetilde{\Upsilon}
\right)
\leq
d_{\mathrm m}^{2}.
\label{eq:comb_memory_rank_bound}
\end{equation}
This recovers the memory-dimension bound of Theorem~\ref{thm:memory_bound} directly from the memory cut of the quantum comb.

\section{Numerical construction details and investigation}
\label{sec:numerical_investigation}

This section gives the numerical details underlying the results in the main text.
We describe the construction and propagation of the interference matrix, the many-body models and numerical resolution, and the additional checks and results supporting the main analysis.

\subsection{Numerical construction of the reference-assisted interference matrix}
\label{subsec:numerical_interference_matrix}

In the numerical analysis, we evaluate the same reference-assisted interference matrix $M_{RS}$ introduced in the main text.
We use the maximally entangled reference-system input
\begin{equation}
\ket{\Phi_{d_S}}_{RS}
=
\frac{1}{\sqrt{d_S}}
\sum_{i=0}^{d_S-1}
\ket{i}_{R}\ket{i}_{S},
\label{eq:numerical_maximally_entangled_input}
\end{equation}
and take the two auxiliary system inputs to be maximally mixed
\begin{equation}
\rho_1
=
\rho_2
=
\frac{I_{d_S}}{d_S}.
\label{eq:numerical_auxiliary_inputs}
\end{equation}
The environment preparation and the two-interaction dynamics are specified in the following subsection. For fixed auxiliary inputs, environment preparation, and dynamics, linearity of the off-diagonal control sector defines a complex-linear map
\[
\mathcal R:
\mathcal L(\mathcal H_S)
\longrightarrow
\mathcal L(\mathcal H_S),
\]
For an input operator $A$, $\mathcal R(A)$ is the off-diagonal system output generated by the pseudo-control circuit.
We include in $\mathcal R$ the factor $1/2$ from the initial $\ket{0}\!\bra{1}_c$ coherence of the control.
The map $\mathcal R$ is only a convenient numerical representation of the off-diagonal block. Using Eq.~\eqref{eq:numerical_maximally_entangled_input}, the reference-assisted interference matrix is
\begin{align}
M_{RS}
&=
\left(
\operatorname{id}_{R}\otimes\mathcal R
\right)
\left(
\ket{\Phi_{d_S}}\!\bra{\Phi_{d_S}}_{RS}
\right)
\nonumber\\
&=
\frac{1}{d_S}
\sum_{i,j=0}^{d_S-1}
\ket{i}\!\bra{j}_{R}
\otimes
\mathcal R
\left(
\ket{i}\!\bra{j}_{S}
\right).
\label{eq:numerical_MRS_choi_expansion}
\end{align}
Defining the unnormalised Choi representation of $\mathcal R$ by
\begin{equation}
J_{\mathcal R}
:=
\sum_{i,j=0}^{d_S-1}
\ket{i}\!\bra{j}_{R}
\otimes
\mathcal R
\left(
\ket{i}\!\bra{j}_{S}
\right),
\label{eq:numerical_choi_definition}
\end{equation}
we therefore have
\begin{equation}
M_{RS}
=
\frac{1}{d_S}
J_{\mathcal R}.
\label{eq:numerical_MRS_normalised_choi}
\end{equation}
The factor $1/d_S$ comes from the normalisation of the maximally entangled input.
It does not change the exact matrix rank, but it fixes the physical scale of the singular values and is therefore retained throughout the numerical analysis. Eq.~\eqref{eq:numerical_MRS_normalised_choi} also gives a useful numerical shortcut.
We do not propagate the reference register explicitly.
Instead, we evaluate the off-diagonal dynamics on the operator basis $\{\ket{i}\!\bra{j}\}$ and assemble $J_{\mathcal R}$ directly.
After the normalisation in Eq.~\eqref{eq:numerical_MRS_normalised_choi}, the resulting matrix is exactly the reference-assisted interference matrix $M_{RS}$.
Although $\mathcal R$ is written in Choi form, it is generally not a quantum channel, and $J_{\mathcal R}$ need not be Hermitian or positive.
Its singular-value spectrum is discussed below.

\subsection{Models and parameters of many-body systems}
\label{subsec:numerical_models}

The first and second system-environment interactions have the same duration $\tau$ and coupling strength $\lambda$.
During either interaction, the joint state evolves according to
\begin{equation}
\mathcal L_{SE}(\rho_{SE})
=
-i
\left[
H_S+H_E+\lambda H_{SE},
\rho_{SE}
\right]
+
\left(
\operatorname{id}_S\otimes\mathcal D
\right)
(\rho_{SE}).
\label{eq:app_interaction_generator}
\end{equation}
Between the two interactions, the system-environment coupling is switched off and $S$ remains idle, while the environment evolves for an interval time $T$ under
\begin{equation}
\mathcal L_E(\rho_E)
=
-i
\left[
H_E,\rho_E
\right]
+
\mathcal D(\rho_E).
\label{eq:app_interval_generator}
\end{equation}
The model-dependent environment Hamiltonians $H_E$ are specified below. The same local depolarising term is used for models A-D, whereas for model E we set $\mathcal D=0$. For all reported simulations, the system $S$ and environment $E$ are three-qubit chains, with
$ 
n=3, d_S=d_E=2^3.
$
We study two system-environment coupling geometries.
The boundary contact couples only the terminal system qubit $S_n$ to the first environment qubit $E_1$
\begin{equation}
H_{SE}^{\mathrm{bdry}}
=
\sum_{\mu\in\{x,y,z\}}
J_c^\mu\,
\sigma_{S,n}^{\mu}
\sigma_{E,1}^{\mu},
\label{eq:app_boundary_contact}
\end{equation}
whereas the ladder contact couples corresponding sites of the two chains
\begin{equation}
H_{SE}^{\mathrm{ladder}}
=
\frac{1}{n}
\sum_{j=1}^{n}
\sum_{\mu\in\{x,y,z\}}
J_c^\mu\,
\sigma_{S,j}^{\mu}
\sigma_{E,j}^{\mu}.
\label{eq:app_ladder_contact}
\end{equation}
We use $J_c^x=J_c^y=0.80$ and $J_c^z=0.30$.
The factor $1/n$ gives the ladder and boundary contacts the same spectral width for the same values of $J_c^\mu$.
The complete numerical protocol is run separately for the two coupling geometries. The environment is initially maximally mixed
\begin{equation}
\rho_E(0)
=
\frac{I_{d_E}}{d_E}.
\label{eq:app_environment_initial_state}
\end{equation}
Models A-D include local depolarisation of the environment
\begin{equation}
\mathcal D(\rho_E)
=
\sum_{j=1}^{n}
\frac{\gamma_j}{4}
\sum_{\mu\in\{x,y,z\}}
\left(
\sigma_j^\mu
\rho_E
\sigma_j^\mu
-
\rho_E
\right),
\label{eq:app_environment_depolarisation}
\end{equation}
where $\sigma_j^\mu$ acts on environment site $j$.
We use the uniform rate
$
\gamma_j=0.05, j=1,\ldots,n.$ This noise is unital and represents local information loss rather than energy relaxation.

We now specify the five environment Hamiltonians. 
For $n=3$, the nonzero Hamiltonians are compared at a common spectral width.
Models C and E set the reference width, while models B and D are multiplied by overall factors $s_B$ and $s_D$.
Each factor rescales the complete Hamiltonian and therefore preserves all coupling and field ratios within that model.

Model A remains the zero-Hamiltonian baseline and is not included in this matching. In the equations below, $X_j$, $Y_j$, and $Z_j$ denote Pauli operators acting on environment site $j$.
Model A has no internal Hamiltonian
\begin{equation}
H_E^{(A)}=0.
\label{eq:app_model_A_hamiltonian}
\end{equation}
Model B is a commuting longitudinal-field Ising chain
\begin{equation}
H_E^{(B)}
=
s_B
\left[
\sum_{j=1}^{n-1}
Z_jZ_{j+1}
+
0.20
\sum_{j=1}^{n}
Z_j
\right],
\label{eq:app_model_B_hamiltonian}
\end{equation}
with
$
s_B=1.40927178.
 $
The rescaling of $s_B$ (as well as other parameters in the following models) is to match the Hamiltonian spectral width of model B to that of models C and E while preserving its internal coupling and field ratio.

Model C is an XXZ-type chain with a local transverse field on the terminal site
\begin{align}
H_E^{(C)}
={}&
1.10
\sum_{j=1}^{n-1}
\left(
X_jX_{j+1}
+
Y_jY_{j+1}
\right)
\nonumber\\
&+
0.35
\sum_{j=1}^{n-1}
Z_jZ_{j+1}
+
0.20
\sum_{j=1}^{n}
Z_j
+
0.20 X_n.
\label{eq:app_model_C_hamiltonian}
\end{align}
The system Hamiltonian has the same form and coefficients as $H_E^{(C)}$, acting instead on the system chain. Model D is a tilted-field Ising chain
\begin{equation}
H_E^{(D)}
=
s_D
\left[
\sum_{j=1}^{n-1}
Z_jZ_{j+1}
+
0.80
\sum_{j=1}^{n}
X_j
+
0.20
\sum_{j=1}^{n}
Z_j
\right],
\label{eq:app_model_D_hamiltonian}
\end{equation}
with
 $
s_D=1.06472678.
 $
As for model B, this overall rescaling matches the Hamiltonian spectral width to that of models C and E without changing the relative interaction and field strengths.

Finally, model E has the same coherent Hamiltonian as model C
\begin{equation}
H_E^{(E)}
=
H_E^{(C)},
\label{eq:app_model_E_hamiltonian}
\end{equation}
but no depolarisation.
The five models therefore provide controlled comparisons between no coherent internal dynamics (A), commuting Ising dynamics (B), exchange-driven XXZ-type dynamics (C), noncommuting tilted-field Ising dynamics (D), and the closed counterpart of model C (E).
In particular, models C and E differ only by the presence or absence of local depolarisation, allowing the effect of dissipation to be separated from the coherent environment dynamics.

For the numerical scans, the first and second interactions use the same duration and coupling strength.
The interaction-time scan covers $0\leq\tau\leq60$ at fixed $\lambda=1$ and $T=1$, with denser sampling over $0\leq\tau\leq15$.
Unless $\tau_0$ is specified directly, we determine it from the boundary-contact interaction-time scan of model C.
Among the sampled points with $\tau\geq0.20$, we identify the widest contiguous plateau attaining the largest resolved rank and choose the sampled point nearest its midpoint.
The same value $\tau_0$ is then used for all models and for both coupling geometries. The coupling-strength scan covers $0\leq\lambda\leq14$ at fixed $\tau=\tau_0$ and $T=1$.
The interval-time scan covers $1\leq T\leq48$ at fixed $\tau=\tau_0$ and $\lambda=1$.
For models C and D, we additionally evaluate two-dimensional $(\lambda,T)$ and $(\tau,T)$ scans.
The heatmaps use 64 points along the interaction-time or coupling-strength axis and 80 interval-time points, with $0\leq\tau\leq15$ for the $(\tau,T)$ scans.

\subsection{Numerical propagation of the interference matrix}
\label{subsec:numerical_propagation}

The off-diagonal control block does not in general evolve as an ordinary density operator, because its ket and bra histories follow different routed system inputs.
If these two histories evolve under Hamiltonians $H_L$ and $H_R$, respectively, an off-diagonal operator $X$ obeys
\begin{equation}
\frac{dX}{dt}
=
-i
\left(
H_LX-XH_R
\right)
+
\left(
\operatorname{id}_S\otimes\mathcal D
\right)(X).
\label{eq:cross_evolution}
\end{equation}
When $H_L=H_R$, this reduces to the usual GKLS evolution. Each interaction of duration $\tau$ is divided into $N=16$ equal slices, with $\Delta t=\tau/N$.
We use the second-order Strang step
\begin{equation}
X_{m+1}
=
\mathcal E_{\Delta t/2}
\left[
U_L(\Delta t)\,
\mathcal E_{\Delta t/2}(X_m)\,
U_R^\dagger(\Delta t)
\right],
\label{eq:cross_strang_step}
\end{equation}
where
\begin{equation}
U_{L,R}(\Delta t)
=
e^{-iH_{L,R}\Delta t},
\qquad
\mathcal E_t
=
\exp
\left[
t\left(
\operatorname{id}_S\otimes\mathcal D
\right)
\right].
\label{eq:cross_step_operators}
\end{equation}
Adjacent dissipative half steps are combined in the implementation.

To evaluate $M_{RS}$ efficiently, we contract the first and second interactions separately rather than propagating the complete pseudo-control circuit at every parameter point.
During the first interaction, the ket history couples the routed system to the environment, while the bra history couples the first auxiliary system input to the same environment.
After tracing out the auxiliary output, we retain the resulting first-interaction tensor $B^{(1)}$.

During the interval, only the environment evolves.
For each model, its evolution is described by
\begin{equation}
\mathsf S_E(T)
=
\exp
\left(
T\mathcal L_E
\right),
\label{eq:interval_superoperator}
\end{equation}
where $\mathcal L_E$ is defined in Eq.~\eqref{eq:app_interval_generator}. The second interaction is treated in the complementary direction.
We contract the final trace over the auxiliary output and the environment backwards through the second interaction, and then contract the second auxiliary input. This gives a second-interaction tensor $C^{(2)}$.

Combining the first interaction, the interval evolution, and the second interaction gives
\begin{equation}
\left[
M_{RS}
\right]_{(i,a),(j,b)}
=
\frac{1}{2d_S}
\sum_{e,f,e',f'}
B^{(1)}_{aef,i}\,
\left[
\mathsf S_E(T)
\right]_{e'f',ef}\,
C^{(2)}_{bj e'f'}.
\label{eq:boundary_tensor_MRS}
\end{equation}
Here $i,j$ label the system input basis and $a,b$ label the retained system output.
The pairs $(e,f)$ and $(e',f')$ are the environment ket and bra indices before and after the interval.
The factor $1/(2d_S)$ contains the initial control coherence and the normalisation of the maximally entangled reference-system input.
Thus the factor $1/2$, which was absorbed into $\mathcal R$ in Sec.~\ref{subsec:numerical_interference_matrix}, is written explicitly in this contraction.

This reduced contraction avoids directly propagating the full off-diagonal control sector at every scan point while remaining exact for the chosen Strang discretisation.
It gives the physically normalised interference matrix $M_{RS}$, whose singular values are used to determine the resolved rank below.

\subsection{Resolved rank and numerical checks}
\label{subsec:numerical_resolved_rank}

Because $M_{RS}$ is generally non-Hermitian, we determine its rank from its singular values
\[
s_1(M_{RS})\geq s_2(M_{RS})\geq\cdots\geq 0.
\]
Very small singular values cannot be reliably resolved in a finite numerical calculation.
We therefore use the resolved rank
\begin{equation}
r_\epsilon
:=
\#\left\{
k:
s_k(M_{RS})>\epsilon
\right\},
\label{eq:numerical_resolved_rank}
\end{equation}
with the fixed absolute cutoff
$ 
\epsilon=10^{-4}.
$ The same physically normalised $M_{RS}$ and the same absolute cutoff are used at every parameter point. 
A decrease in $r_\epsilon$ means that fewer singular values exceed the same cutoff. The resolved rank should be distinguished from the exact rank appearing in Theorem~\ref{thm:memory_bound}.
For an exactly known matrix
\[
r_\epsilon
\leq
\operatorname{rank}(M_{RS}),
\]
because applying a nonzero threshold can only discard singular values.
We therefore report
\begin{equation}
d_{\mathrm m,\epsilon}
=
\max\left\{
1,
\left\lceil
\sqrt{r_\epsilon}
\right\rceil
\right\}.
\label{eq:numerical_resolved_memory_dimension}
\end{equation}
This is the memory-dimension bound resolved at the chosen cutoff $\epsilon$.

We next check that the numerical errors relevant to the reported calculations are smaller than this cutoff.
We first check the reduced contraction in Eq.~\eqref{eq:boundary_tensor_MRS} against direct propagation of the complete cross circuit.
Because the full calculation grows rapidly with system size, this comparison is performed for representative one- and two-qubit configurations, including dissipative dynamics and both boundary and ladder contacts.
We find no discrepancy larger than the numerical precision used in these tests.

We next examine convergence with respect to the Strang discretisation.
The reported calculations use $N=16$ slices for each interaction, and we compare them with the corresponding $N=8$ calculations at representative parameter points.
For a second-order Strang scheme, we estimate the error of the $N=16$ calculation as
\begin{equation}
\delta_{\mathrm{num}}
=
\max_x
\frac{
\left\|
M_{RS}^{(16)}(x)
-
M_{RS}^{(8)}(x)
\right\|_{\mathrm{op}}
}{2^2-1},
\label{eq:numerical_richardson_error}
\end{equation}
where the maximum is taken over the tested parameter points.
For the reported calculations,
$\delta_{\mathrm{num}}<10^{-4}=\epsilon$
at all tested points. This is a convergence estimate, not a rigorous a priori error bound. 
We use it only to check that the $N=16$ discretisation error is smaller than the cutoff $\epsilon$.

Finally, we perform four null tests.
We set either the duration or the coupling strength of the first or second interaction to zero, one at a time.
In each case the resolved rank returns to at most the rank-one baseline
\begin{equation}
r_\epsilon\leq 1.
\label{eq:numerical_null_test}
\end{equation}
Thus, removing either interaction restores the rank-one baseline.

\subsection{Coupling-geometry dependence beyond the main-text cuts}
\label{subsec:numerical_coupling_geometry}

The main text uses the boundary contact to examine the temporal dependence of the resolved rank, while the comparison between boundary and ladder contacts focuses on the coupling-strength dependence.
Here we test whether the temporal behaviour also persists when every system site is coupled directly to the corresponding environment site.

Fig.~\ref{fig:app_ladder_temporal} shows the interaction-time and interval-time scans for the ladder contact.
The same broad temporal behaviour seen for the boundary contact remains visible.
At short and intermediate interaction times, the two interactions resolve an increasing number of singular directions.
For the dissipative models A-D, the resolved rank is suppressed at longer interaction or interval times, whereas the closed model E retains a substantially larger and more strongly fluctuating rank.
Thus, the basic temporal competition between coherent access to the environment and dissipative loss does not rely on the single-link boundary contact.

\begin{figure}[htbp]
    \centering

    \begin{minipage}{0.485\linewidth}
        \centering
        \includegraphics[width=\linewidth]
        {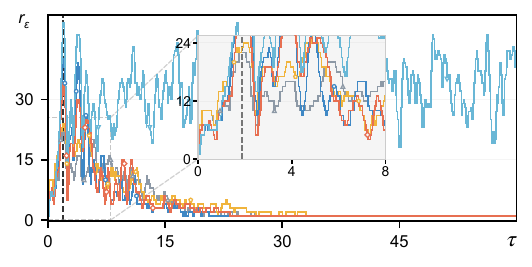}
        \textbf{(a)}
    \end{minipage}
    \hfill
    \begin{minipage}{0.485\linewidth}
        \centering
        \includegraphics[width=\linewidth]
        {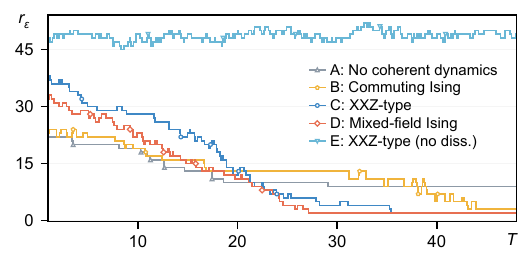}
        \textbf{(b)}
    \end{minipage}

    \caption{Temporal dependence for the ladder contact.
    (a) Resolved rank $r_\epsilon$ versus the common interaction time $\tau$ at $\lambda=1$ and $T=1$.
    The vertical dashed line marks the working time $\tau_0$ selected from the boundary-contact scan and used for both geometries.
    (b) Resolved rank versus interval time $T$ at $\tau=\tau_0$ and $\lambda=1$.
    All panels use the fixed absolute resolution $\epsilon=10^{-4}$.
}
    \label{fig:app_ladder_temporal}
\end{figure}

The stronger geometry dependence appears when the coupling strength is varied.
To determine whether the one-dimensional behaviour in Fig.~\ref{fig:numerical_coupling_geometry} is specific to the chosen interval time, Fig.~\ref{fig:app_lambda_wait_geometry} shows the full $(\lambda,T)$ landscapes for models C and D.
For each model, the boundary and ladder panels use the same colour scale so that their resolved ranks can be compared directly.

For the boundary contact, the larger resolved ranks are concentrated in relatively restricted regions of coupling strength and short interval time.
At stronger coupling, these regions narrow or give way to lower-rank regions.
The ladder contact retains nontrivial resolved rank over a broader range of coupling strengths and shows additional high-rank structure at larger $\lambda$.
The detailed pattern depends on the environment Hamiltonian, but the distinction between local boundary access and distributed ladder access persists over a finite range of interval times rather than only along the one-dimensional cut shown in the main text.

\begin{figure}[htbp]
    \centering

    \begin{minipage}{0.48\linewidth}
        \centering
        \includegraphics[width=0.72\linewidth]
        {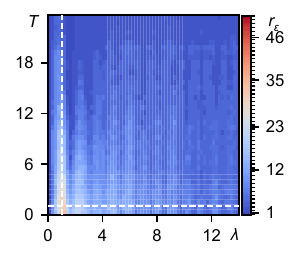}\par
        \textbf{(a)} Model C, boundary
    \end{minipage}
    \hfill
    \begin{minipage}{0.48\linewidth}
        \centering
        \includegraphics[width=0.72\linewidth]
        {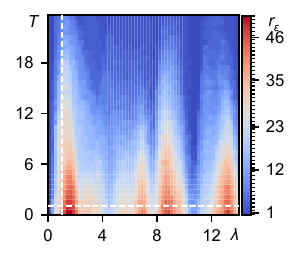}\par
        \textbf{(b)} Model C, ladder
    \end{minipage}

    \vspace{0.8em}

    \begin{minipage}{0.48\linewidth}
        \centering
        \includegraphics[width=0.72\linewidth]
        {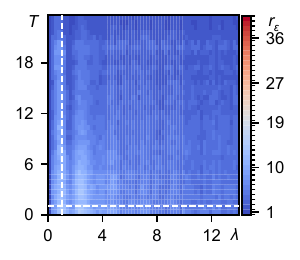}\par
        \textbf{(c)} Model D, boundary
    \end{minipage}
    \hfill
    \begin{minipage}{0.48\linewidth}
        \centering
        \includegraphics[width=0.72\linewidth]
        {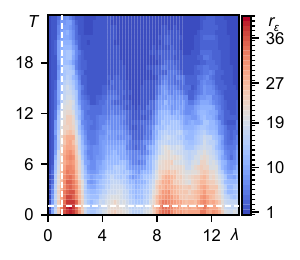}\par
        \textbf{(d)} Model D, ladder
    \end{minipage}

    \caption{Coupling-strength--interval-time landscapes for the two coupling geometries. 
    Resolved rank $r_\epsilon(\lambda,T)$ for (a,b) model C and (c,d) model D. 
    The left column uses the boundary contact and the right column the ladder contact.
    In each model, the two geometries use the same integer colour scale.
    The interaction time is fixed at $\tau=\tau_0$, and the white dashed lines mark $\lambda=1$ and $T=1$.
    All panels use $\epsilon=10^{-4}$.
}
    \label{fig:app_lambda_wait_geometry}
\end{figure}

We display the resolved-rank landscapes rather than separate maps of $d_{\mathrm m,\epsilon}$ because the latter are obtained directly from
$d_{\mathrm m,\epsilon}
=
\max\{1,\lceil\sqrt{r_\epsilon}\rceil\}$
and therefore contain a coarser version of the same information.

Together, these comparisons separate two effects.
The basic temporal suppression of resolved memory under dissipation persists for both coupling geometries, whereas access to the many-body environment at strong coupling depends strongly on how the system is connected to it.

\subsection{Singular-value structure of the resolved rank}
\label{subsec:numerical_singular_structure}

The resolved rank compresses the singular-value spectrum of $M_{RS}$ into a single integer.
To show what the resulting rank changes represent, we examine the spectra directly at several interval times.
Fig.~\ref{fig:app_CE_singular_spectra} compares models C and E for the boundary contact. These two models have the same coherent environment Hamiltonian and differ only by the presence of local depolarisation.

For model C, increasing the interval time shifts an increasing number of singular values below the fixed resolution $\epsilon=10^{-4}$.
The corresponding decrease in $r_\epsilon$ is therefore the discrete record of continuous changes in the singular-value spectrum.
By contrast, the spectra of the closed model E remain much more strongly overlapping as $T$ is varied.
Their remaining changes arise from coherent finite-size dynamics rather than a systematic dissipative suppression of the singular values.
The comparison therefore resolves at the spectral level the different interval-time behaviour of models C and E discussed in the main text.

\begin{figure}[htbp]
    \centering

    \begin{minipage}{0.48\linewidth}
        \centering
        \includegraphics[width=0.78\linewidth]
        {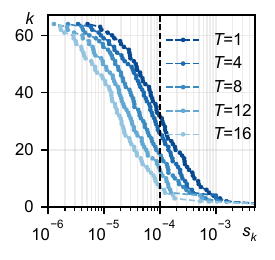}\par
        \textbf{(a)} Model C
    \end{minipage}
    \hfill
    \begin{minipage}{0.48\linewidth}
        \centering
        \includegraphics[width=0.78\linewidth]
        {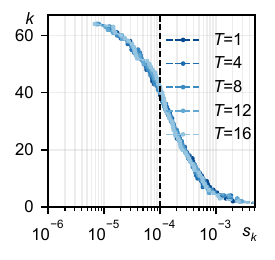}\par
        \textbf{(b)} Model E
    \end{minipage}

    \caption{Singular-value spectra underlying the interval-time dependence. Singular-value index $k$ versus $s_k(M_{RS})$ for (a) dissipative model C and (b) closed model E, at selected interval times $T=1,4,8,12,16$.
    The boundary contact is used with $\tau=\tau_0$ and $\lambda=1$. 
    The vertical dashed line marks the fixed absolute resolution $\epsilon=10^{-4}$. 
    Singular values to the right of this line contribute to $r_\epsilon$.
}
    \label{fig:app_CE_singular_spectra}
\end{figure}

These spectra also clarify the finite-resolution meaning of the numerical results.
A change in $r_\epsilon$ occurs when a singular direction crosses the fixed resolution scale; it is not a discontinuity of the underlying interference matrix.
Conversely, a singular direction that falls below $\epsilon$ is unresolved at the adopted numerical resolution.
Its absence from $r_\epsilon$ does not by itself imply that the corresponding microscopic information has been completely erased.

\end{document}